\documentclass{article}
\usepackage{amsmath,graphicx}

\usepackage{amsthm}
\usepackage{amssymb}
\usepackage{fancyhdr}
\usepackage{xypic}
\usepackage{hyperref}
\theoremstyle{plain}
\newtheorem{theorem}{Theorem}
\newtheorem{proposition}{Proposition}
\newtheorem{corollary}{Corollary}
\newtheorem{lemma}{Lemma}
\theoremstyle{definition}
\newtheorem{definition}{Definition}

\newtheorem{example}{Example}

\def\supp{\textrm{supp }}

\def\cl{\textrm{cl }}

\def\setint{\textrm{int }}

\fancypagestyle{plain}{
\fancyhf{}
\lhead{}
\rhead{}
\chead{}
\cfoot{}
}

\title{Topological and geometric characterization of synthetic aperture sonar collections}
\author{Michael Robinson$^1$ \and Zander Memon$^2$ \and Maxwell Gualtieri$^3$}
\date{%
  $^1$Mathematics and Statistics, American University, Washington, DC, USA, michaelr@american.edu\\%
  $^2$Department of Economics, University of California, Davis, CA, USA, zhmemon@ucdavis.edu\\%
  $^3$McCormick School of Engineering, Northwestern University, Evanston, IL, USA, maxgualtieri2026@u.northwestern.edu\\[0.5cm]%
  \today
}

\begin{document}
\maketitle

\begin{abstract}
This article explores the theoretical underpinnings of---and practical results for---topological and geometric features of data collected by synthetic aperture sonar systems.
We prove a strong theoretical guarantee about the structure of the space of echos (the signature space) that is relevant for classification,
and that this structure is measurable in laboratory conditions.
The guarantee makes minimal assumptions about the sonar platform's trajectory,
and establishes that the signature space divides neatly into topological features based upon the number of prominent echos and their placement,
and geometric features that capture their corresponding sonar cross section.  
\end{abstract}

\tableofcontents

\section{Introduction}
\label{sec:intro}

The performance of many classification methods is determined by the geometry of the space of possible objects being classified.
Recent work on topological data analysis (TDA) establishes that the topology of this space is sometimes meaningful as well.
In the case of sonar object classification, most methods require the raw sonar echos to be formed into images before classification is attempted.
While this helps isolate objects spatially, sonar responses are not quite as localized as they are in optical systems, so the image domain is somewhat lossy.
However, if we attempt to classify objects in the raw sonar echo domain, then the geometry and topology of the \emph{space of echos}---rather than the space of all objects---becomes available as a feature for classification.
The space of echos is related to the object via physical laws, but the interaction of different physical acoustic processes complicates the interpretation of this space.
Moreover, the topology and geometry of the space of echos can show the impact of the parameters of the collection system as much as the object being classified.
This paper explores the theoretical underpinnings of---and practical results for---synthetic aperture sonar systems.
We demonstrate that there are theoretical guarantees about the structure of the space of echos that are relevant for classification,
and that these properties are measurable in laboratory conditions.

Perhaps most important is the recognition that the structure of the space of echos itself is a useful feature for classification.
Since we wish to isolate the response of the object being classified from the collection system,
we seek representations of the space of echos that are invariant with respect to (at least) the trajectory of the sensor.
Specifically, rather than forming an image, and rather than treating the echos individually,
we study the space of all echos under certain equivalences induced by changes in trajectory.
This explores situations when classification performance is likely to be at its worst, and so can provide \emph{realistic lower bounds on performance}.
Paradoxically, if a classification system can work with trajectory invariant representations of the space of echos,
then it can continue to perform well in severe conditions, because it is completely insensitive to uncertainties in the trajectory that would normally cause blurring or other image-related distortions.

Because of the physical constraints of the collection system and physical acoustics, the space of echos exhibited by an object will tend to be a manifold or an immersed manifold.
Since most artificial objects tend to excite very definite propagation phenomena, we prove under generic conditions that the space of echos has a very definite dimension.
Furthermore, the propagation behaviors determined by physical acoustics tend to be isolated in space, time, or frequency,
a fact which leads to a definite topological constraint on the space of echos.
The estimates of topological features can then be applied to experimental sonar data to identify these different propagation phenomena for specific objects.

Ultimately, we show that the tentative results of good classification performance using persistent homology (a trajectory invariant geometric-topological feature) in \cite{sonarspace} is a consequence of the constraints on the space of echos exerted by physical acoustics.
This turns out to not depend strongly upon the exact nature of the physical laws, but more upon their basic properties.
Therefore, our assessment is that TDA-based classifiers have a strong basis for broad application in sonar classification problems.

\subsection{Contributions}

\label{sec:contribution}

We show both theoretically (Section \ref{sec:theory}) and experimentally (Section \ref{sec:verification}) that there are strong genericity properties for synthetic aperture sonar.
These are governed by a detailed theoretical characterization of the structure of sonar signatures under minimal assumptions (Section \ref{sec:signal_model}) about the sonar platform's trajectory.
The structure that is present corresponds directly to physical acoustic properties of the target in a straightforward way.

Our main result (Theorem \ref{thm:echoes_bouquet_spheres}) shows that the structure of the space of echos divides neatly into \emph{topological features} based upon the number of prominent echos and their placement, and \emph{geometric features} that capture their sonar cross section.
This characterization is dependent upon a \emph{genericity} assumption, with the observation that artificial objects will tend to be rather non-generic.
As a consequence, \emph{unusual topological features in the space of sonar echos implies that the object being imaged is of artificial origin}.

This article establishes the following new results:
\begin{enumerate}
\item A new mathematical definition of a space of functions that represent generic sonar signatures with the possibility of isolated prominent echos (Definition \ref{def:generic_with_support}),
\item A complete characterization of the topology, and a bound on the geometry, of these signatures (Theorem \ref{thm:echoes_bouquet_spheres}),
\item A weaker, but computationally accessible version of the purely topological information (Corollary \ref{cor:homology}),
\item Specialization of this characterization to circular synthetic aperture sonar (CSAS) (Proposition \ref{prop:csas_direct}) and to sets of point scatterers (Proposition \ref{prop:point_injectivity}),
\item A bound on the resulting persistence diagrams (Proposition \ref{prop:sonar_pd}), and
\item Verification of the above theoretical results in simulation (Section \ref{sec:sim_expt}) and in laboratory conditions (Section \ref{sec:expt}).
\end{enumerate}

\subsection{Historical context}
\label{sec:historical}

From a theoretical perspective, the separation of prominent from diffuse scatterers is hardly new (see for instance \cite{goehle2022enveloped,Schupp-OmidDaniel2016Coas,Wilbur1993ApplicationOW}),
but the fact that it separates topological from geometric information appears to be novel.
It has been widely acknowledged that while topological features of an embedded space may be robust in persistent homology,
the geometry of that space impacts the persistent homology in a significant way.
The geometric information in a sonar signature is important for classification \cite{Eremeev2019AnAF}, so it is important to consider whether (and if so, how) persistent homology captures this information.
This article shows that if one can separate the prominent echos from those that are more diffuse,
then the impact of different physical acoustic processes on the resulting persistent homology features will be distinct.

For instance, for CSAS collections, it is expected that the resulting signature will be periodic in look angle.
In this setting, one can apply the theoretical characterization of the persistent homology of quasiperiodic (rationally independent periods) signatures discussed in several articles \cite{Gakhar2021SlidingWP,Perea2016PersistentHO,Perea2013SlidingWA}.
It is also known that curvature tends to impact the short bars of persistent homology,
\cite{turkes2022effectiveness,Adams2021TopologyAT,Bubenik2019PersistentHD,Virk20171DimensionalIP}.
The clear examples of curvature effects in the literature tend to be rather symmetric spaces that differ from practical sonar signatures (for instance \cite{Adams_2019,Adamaszek2014NerveCO,Adamaszek2015TheVC}).
While it would be useful to have precise bounds on the impact of geometric information on persistent homology,
for the purposes of this article, the intuition that curvature is relevant will suffice.
As a result, the quality of the experimental and simulated data is paramount, and our presentation here is apparently novel in that regard.

\section{Signal model}
\label{sec:signal_model}

A sonar collection can be represented as a function $u: M \to \mathbb{R}^n$,
where the codomain, $\mathbb{R}^n$, represents $n$ real samples collected by the receiver.
Complex samples, in $\mathbb{C}^n$, can also be modeled.
Except for Section \ref{sec:point_scatterers}, the results proven (theoretically and experimentally) in this article do not rely upon the algebra or geometry inherent in $\mathbb{C}^n$,
though some of the signal models are more succinctly expressed using complex numbers,
so we will use either real or complex samples as appropriate.

Each point of the domain $M$ of $u$ represents a distinct observation that can be collected by the system.
It is mathematically useful (and experimentally accurate) to assume that $M$ is a manifold without boundary (a \emph{manifold} without further qualification in what follows).
Although many of our examples will consider compact or connected $M$, neither are generally required.

The dimension of $M$ corresponds to the number of free parameters necessary to specify an observation,
for instance the sensor look angle, range, or center frequency, among many other possibilities.
It is useful to think of $M$ as representing possible configuration states of the transmitter.
The configuration states can be various; for instance, 
\begin{itemize}
\item If the sonar platform makes a circular orbit around a target, then $M=S^1$ is a circle that parameterizes the look angle,
\item If the sonar platform moves along a linear trajectory past a target, then $M=[0,1]$ parameterizes the sensor's position along the track in arbitrary units, or
\item If the sonar platform scans all azimuth and elevation angles around a free-floating target, then $M=S^2$ is the sphere, parameterizing all possible azimuth and elevation angles, yet
\item Many other possibilities can be imagined.
\end{itemize}
Both monostatic and multistatic sonar systems can be modeled in this way, though the data analyzed in this paper is entirely monostatic.

There may be several possible ways of formulating the same collection as a function.
For instance one might consider a CSAS collection to be written as a function $u: S^1 \times \mathbb{R} \to \mathbb{C}$,
in which the domain specifies look angle and range,
while the codomain specifies the complex echo received from that look angle and range.
However, one could also write the same collection as
$u: S^1 \to \mathbb{C}^n$
in which the domain consists of look angle only,
and the range samples index the dimensions of the codomain.
The theoretical results proven in this article apply equally well to both formulations,
and we will freely choose whichever formulation is more convenient for the task at hand.

For the purposes of this article, assume that $u$ can be written as a sum
\begin{equation*}
  u(x) = v(x) + n(x)
\end{equation*}
of a smooth function $v$ and a noiselike $n$.
Our theoretical attention will be squarely upon $v$,
because it has been established that under conditions of high signal-to-noise ratio,
$n$ will tend to have limited effect on the topology of $u$ as compared to $v$ \cite{Carlsson_2009,Ghrist_2008,Niyogi_2008,Chazal_2006}.
That said, the effect of $n$ \emph{is of extreme practical importance}, and so will be examined under simulation in Section \ref{sec:sim_expt} and laboratory experimental conditions in Section \ref{sec:expt}.

\subsection{Prominent echos}

Many sonar collections exhibit clear boundaries between different classes of acoustic phenomena.
In order to formalize what this means,
we will split the domain $M$ into various regions on which a particular phenomenon is dominant.
From a survey of the literature,
acoustics researchers tend to confine their attention to situations when these regions happen to be \emph{topological disks}.
A \emph{topological disk} (or simply a \emph{disk}) $D$ is a manifold with boundary that is diffeomorphic to the closed unit ball in Euclidean space,
\begin{equation*}
  \left\{x \in \mathbb{R}^{d+1} : \|x\| \le 1 \right\}.
\end{equation*}
Observe that $D$ has dimension $d+1$ according to the usual invariance of dimension.
A topological disk is a disk in topology only.
It need not be ``round'' or otherwise reflect the geometry of the unit ball in Euclidean space.

For instance, in \cite{Marston_2014}, the authors showed that different echo processes can be modeled by decomposing the raw data (in the form of a function $u$ as described above) into a superposition of cones in angle-wavenumber space.
Since cones are diffeomorphic to disks, our model treats this setting as a special case.
At least under controlled environments, empirical CSAS data often show that different acoustic phenomena are isolated on disks.
For instance, \cite[Figs 5A, 7A]{Marston_2011} show CSAS collections of two targets in which the raw data exhibit distinct echos that are isolated in time and look angle.
These are by no means isolated examples; the laboratory dataset we analyze in Section \ref{sec:expt} also exhibits strongly isolated echos.

This article makes claims about \emph{generic} properties of sonar data.
\begin{definition}
  \label{def:generic}
  A subset $A$ of a topological space $X$ is \emph{dense} if its closure is all of $X$.
  By a \emph{generic} subset $A \subseteq X$, we mean that $A$ is open and dense.
\end{definition}

Many properties of smooth functions hold not just for some functions,
but for generic ones.
This means that they are rather common.
We also need to handle the case where the function is generic subject to additional constraints,
as captured by the following definition.

\begin{definition}
  \label{def:generic_with_support}
  The space of smooth functions whose domain is $M$ and codomain is $N= \mathbb{R}^n$ will be denoted $C^\infty(M,N)$.
  
  Consider a finite set of pairwise disjoint disks $\mathcal{D}$ in a manifold $M$. 
  The space of smooth functions with support in the union $\cup \mathcal{D}$ will be denoted by $C_{\mathcal{D}}^\infty(M,N)$.
  Since $C_{\mathcal{D}}^\infty(M,N)$ is a subspace of $C^\infty(M,N)$,
  we assume that $C_{\mathcal{D}}^\infty(M,N)$ has the subspace topology inherited from $C^\infty(M,N)$.

  Most of our theoretical results will hold for all elements in an open and dense subset $G \subseteq C_{\mathcal{D}}^\infty(M,N)$.
  Briefly, we will say that a smooth function $u: M\to N$ in $G$ is a function that is \emph{generic with support in $\mathcal{D}$}.
\end{definition}

Intuitively, each disk in $\mathcal{D}$ contains a distinct acoustic phenomenon that is isolated from other acoustic phenomena.

\begin{definition}
  \label{def:prominent_echo}
  If $v$ is a smooth function in $C_{\mathcal{D}}^\infty(M,\mathbb{R}^n)$ corresponding to a sonar data collection,
  we will call $v$ a \emph{signal map}.
  We will call the image $v(M) \subseteq N$ the \emph{signature space} for this sonar collection.
  Each restriction $v_D$ of $v$ to a disk $D \in \mathcal{D}$ will be called an \emph{isolated prominent echo},
  or simply a \emph{prominent echo}.
  We will call the quantity $\sigma_k = \|v|_D(t)\|_\infty$,
  which is the maximum value of $v$ on $D$, the \emph{(sonar) cross section} of the prominent echo.
\end{definition}

A given signal map $f \in C^\infty(M,N)$ may be an element of $C_{\mathcal{D}}^\infty(M,N)$ for several different possibilities of $\mathcal{D}$.
The following nesting property suggests that we take the collection $\mathcal{D}$ with the smallest disks that is applicable to our situation.

\begin{proposition}
  Suppose that $\mathcal{D}_1$ and $\mathcal{D}_2$ are two finite collections of pairwise disjoint disks in a manifold $M$ such that
  the union $\cup \mathcal{D}_1$ is a proper subset of the union $\cup \mathcal{D}_2$.
  Under this condition, we have that if $N = \mathbb{R}^n$ for some $n$,
  \begin{enumerate}
    \item $C_{\mathcal{D}_1}^\infty (M,N) \subset C_{\mathcal{D}_2}^\infty (M,N)$, and
    \item $C_{\mathcal{D}_1}^\infty (M,N)$ is not an open subset of $C_{\mathcal{D}_2}^\infty (M,N)$.
  \end{enumerate}
\end{proposition}
Briefly, $C_{\mathcal{D}_1}^\infty (M,N)$ cannot be a generic subset of $C_{\mathcal{D}_2}^\infty (M,N)$,
and moreover every generic subset of $C_{\mathcal{D}_2}^\infty (M,N)$ has prominent echos incompatible with $\mathcal{D}_1$.
\begin{proof}
  To establish the first statement, suppose that $f \in C_{\mathcal{D}_1}^\infty (M,N)$.
  This means that the support of $f$ is contained in $\cup \mathcal{D}_1 \subset \cup \mathcal{D}_2$.
  As a result, $f \in C_{\mathcal{D}_2}^\infty (M,N)$.

  On the other hand,
  we will show that for every $\epsilon > 0$, the ball of radius $\epsilon$ centered on $f$ in $C_{\mathcal{D}_2}^\infty (M,N)$ contains a function outside of $C_{\mathcal{D}_1}^\infty (M,N)$.
  This claim means that $C_{\mathcal{D}_1}^\infty (M,N)$ is not an open subset of $C_{\mathcal{D}_2}^\infty (M,N)$.

  Since $\mathcal{D}_1$ is a finite collection of closed subsets of $M$, this means that $\cup \mathcal{D}_1$ is also closed.
  As a result, the complement $M \backslash \cup \mathcal{D}_1$ is open.
  Since each disk in $\mathcal{D}_2$ is of the same dimension as $M$, this means that its interior is open.
  Hence, the union $\cup \mathcal{D}_2$ has nonempty interior $\setint \cup \mathcal{D}_2$.
  As a result, the set
  \begin{equation*}
    U := \left(\setint \cup \mathcal{D}_2\right) \backslash \left(\cup \mathcal{D}_1\right) = \left(\setint \cup \mathcal{D}_2\right) \cap \left(M \backslash \cup \mathcal{D}_1\right)
  \end{equation*}
  is open.
  Since $M$ is a manifold, there is a partition of unity $\lambda_{\mathcal{V}}$ subordinate to the cover $\mathcal{V}$ generated by $\mathcal{D}_2 \cup \{U\}$.
  In this partition of unity, the function $\lambda_U$ is supported on $U$, and hence is an element of $C_{\mathcal{D}_2}(M,N)$ but is not an element of $C_{\mathcal{D}_1}(M,N)$.
  
  On the other hand, for every $\epsilon > 0$, by construction of the partition of unity,
  \begin{equation*}
    \| \epsilon \lambda_U \|_\infty < \epsilon.
  \end{equation*}
  As a result, $f + \epsilon \lambda_U$ is in the ball of radius $\epsilon$ centered at $f$ in $C_{\mathcal{D}_2}^\infty (M,N)$,
  but since it takes a nonzero value outside of $\cup \mathcal{D}_1$, it cannot be an element of $C_{\mathcal{D}_1}^\infty (M,N)$.
\end{proof}

It is especially clear that the prominent echos isolate different acoustic phenomena in the case of specular reflections.

\begin{figure}
  \begin{center}
    \includegraphics[width=4in]{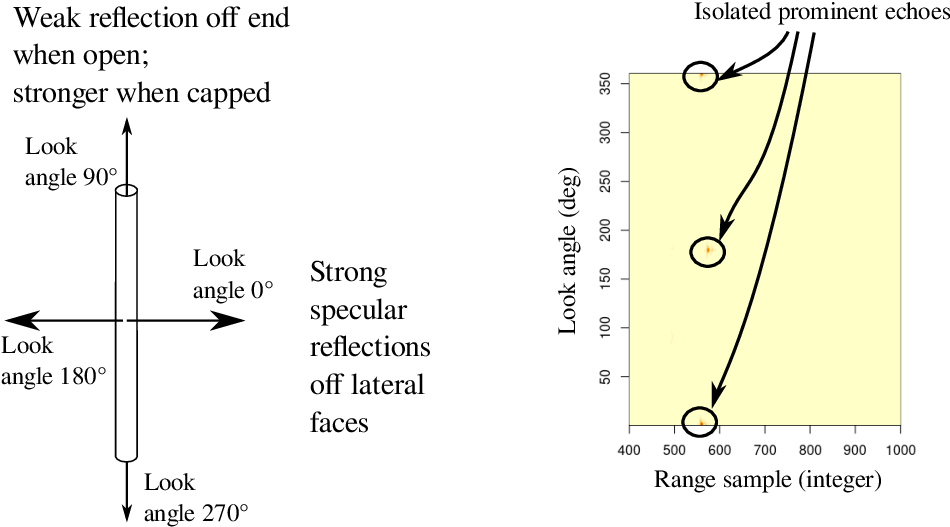}
    \caption{Specular reflections arising from broadside illumination of a copper pipe form prominent echoes.}
    \label{fig:pipe_specular}
  \end{center}
\end{figure}

\begin{example}
Consider the case of a monostatic CSAS image of a narrow copper pipe collected with the AirSAS system \cite{cowenairsas}.
When the incident ray from the transmitter does not form a right angle with the axis of the pipe,
most of the reflected wave from the transmitter does not return to the receiver.
As a result, the receiver reports low received signal.
Conversely, when the incident ray does form a right angle with the axis of the pipe,
a specular reflection returns a larger fraction of the signal to the receiver.
As a result, the specular reflection forms an isolated prominent echo.

Laboratory sonar data corresponding to this situation is shown in Figure \ref{fig:pipe_specular},
where the vertical axis is the look angle of the sonar system and the horizontal axis is the range (further details about the collection system, object geometry, \emph{et cetera} appear in Section \ref{sec:expt_data}).
The brightness of each pixel corresponds to the received signal strength.
Because the look angle can be represented by the circle $S^1$
and the range can be represented by a real number,
the resulting signal map is therefore a function $S^1\times \mathbb{R} \to \mathbb{R}$.

The two prominent echoes highlighted in the figure are specular reflections from opposite sides of the copper pipe.
The disks marked are possible choices for two support sets $D_1,D_2 \subset S^1\times \mathbb{R}$.
Notice that the disk on the top and bottom of the plot are actually two halves of the same disk because the look angle wraps around.
According to Definition \ref{def:prominent_echo},
the maximum received signal level within each disk is the cross section of the corresponding prominent echo.
\end{example}

\subsection{Limitations}
\label{sec:limitations}

The reader is cautioned that for linear propagation media,
the use of $C_{\mathcal{D}}^\infty(M,N)$ with any nontrivial choice of $\mathcal{D}$ is at best an approximation.
Indeed, unless there are nonlinearities in the medium or the target, the Paley-Wiener theorem precludes compactly supported radiation patterns in the far field for targets localized in space.

Most practical scenarios will feature signal maps that take \emph{small but nonzero} values outside of the prominent echos instead of the value $0$ as required by Definition \ref{def:prominent_echo}.
Nevertheless, the position of this article is that neglecting these small values is of considerable practical explanatory value, and is akin to taking the first term in an asymptotic expansion.

With this caveat in mind, many signals can be approximated by functions in $C_{\mathcal{D}}^\infty(M,N)$.
For instance, functions in $C_{\mathcal{D}}^\infty(M,N)$ can arise as sparse linear combinations of compactly supported wavelets.
However,
it is not the case that every smooth function $f: M \to N$ is contained in some $C_{\mathcal{D}}(M,N)$ space.

\begin{example}
  \label{eg:counter_signal}
Suppose that $f$ is a function that is nonzero and constant on the circle $S^1$.
In such a case, the support of $f$ is the entire circle.
It is a standard fact that the circle $S^1$ is not contractible \cite[Thm. 1.7]{Hatcher_2002}.
Since $C_{\mathcal{D}}^\infty(S^1,N)$ requires each component of the support of $f$ to be contained in a disk,
it is impossible for any such set of functions to contain $f$ since $S^1$ is not contractible.
\end{example}

Although Example \ref{eg:counter_signal} is an extreme case from a practical standpoint,
it may happen that the only choice for $C_{\mathcal{D}}(M,N)$ to contain a signal map is one in which $\mathcal{D} = \{M\}$.
This situation may arise if a target consists of complicated, overlapping distributed scatterers.
As will be shown in Theorem \ref{thm:echoes_bouquet_spheres},
although such a target's signal map may have rich geometric information,
it will contain limited topological information. 

In this situation a finer assumption usually comes into play.
Because of Sard's theorem \cite[Thm. 10.7]{Lee_2000}, 
critical points of an arbitrary smooth function $u$ are rare; their complement forms a generic set.
In most physical situations, the critical points with extremal critical values will be isolated.

\begin{proposition}
  In a signal map with isolated critical points, each prominent echo contains at least one critical point.
\end{proposition}
\begin{proof}
  Let $v$ be the signal map under study.
  Since each critical point is assumed to be isolated,
  this also implies that the value of $v$ on at least one point within each prominent echo is nonzero.
  Such a point is within the interior of each prominent echo's support, by definition.
  Prominent echo are supported on disks, each of which are compact.
  This means that the extrema of $\|v\|$, a continuous, nonzero function, are attained within each prominent echo's support.
  Each of these extrema correspond to critical points.
\end{proof}

%
%
%

\section{The structure of signature space}
\label{sec:theory}

Given that a sonar signature is modeled by a signal map $v \in C_{\mathcal{D}}^\infty(M,\mathbb{R}^n)$, it is natural to study
properties of the image of $v$, namely the set $v(M) \subset \mathbb{R}^n$.  We will call $v(M)$ the \emph{signature space} for $v$.
The topology of the signature space $v(M)$ is inherited as a subspace from that of
$\mathbb{R}^n$, and therefore it is useful to consider topological
properties of $v(M)$ as a proxy for properties of $v$.

A \emph{reparameterization} of $v$ consists of the composition $v \circ \phi$, where $\phi$ is a diffeomorphism $M \to M$.
Notice that reparameterizing $v$ does not change $v(M)$ at all.
The signature space $v(M)$ is therefore invariant with respect to changes in the speed\footnote{This assumes that the differences in the speed of the sonar platform are well below the speed of sound, otherwise the Doppler effect will modulate $v$, which can change the signature space $v(M)$.} at which a sonar platform moves.
Moreover, \cite[Prop. 11]{Robinson_constrank} implies that more general trajectory distortions also leave the signature space unchanged.
As such, the signature space is a \emph{trajectory invariant representation} of the signal map $v$.

It is usually not the case that $v(M)$ is a manifold,
because $v$ will usually fail to be an embedding.
The most common reason is that $v$ fails to be injective.
Each pair of points $x_1,x_2 \in M$ where $v(x_1) = v(x_2)$ yield \emph{self-intersections} in $v(M)$.
Under generic conditions, Proposition \ref{prop:self_intersection_only_zero} establishes that self-intersections do not occur \emph{within} a prominent echo.
This helps us to characterize the topological structure of $v(M)$ in Theorem \ref{thm:echoes_bouquet_spheres}.
We specialize these results to CSAS in Section \ref{sec:embedded_csas}.

It is a standard fact that a disk $D$ is a contractible topological space,
whose interior is homeomorphic to $\mathbb{R}^{d+1}$.
The boundary $\partial D$ of $D$ is diffeomorphic to the sphere $S^d$ of dimension $d$.
If we collapse $\partial D$ to a point we obtain the space $D / \partial D$, which is homeomorphic to the sphere $S^{d+1}$ (see Figure \ref{fig:sphere_forming}).
Since acoustic phenomena are localized to disks,
Theorem \ref{thm:echoes_bouquet_spheres} in this section shows that under broad theoretical conditions,
the signature space will decompose into spheres.

\subsection{The general case}

The most general setting applies no restrictions to the domain $M$.

\begin{proposition}
  \label{prop:self_intersection_only_zero}
  Let $M$ be a smooth manifold and 
  assume that $n > (2 \dim M)$.
  For a generic $v$ in $C_{\mathcal{D}}^\infty(M,\mathbb{R}^n)$,
  the self-intersections in the signature space of $v$ only occur at points where $v(t) = 0$. 
\end{proposition}
\begin{proof}
  Suppose that $D \in \mathcal{D}$, which is by assumption a disk in $M$.
  The interior of $D$ is an open subset of $M$, and is therefore a manifold of dimension $\dim M$ on its own.
  Therefore the Whitney immersion theorem \cite[Thm. 10.8]{Lee_2000} rules out self-intersections within $D$
  under the assumption that $n > 2 \dim M$ for a generic subset of $C^\infty(D,\mathbb{R}^n)$.
  The same reasoning applies to the open subset\footnote{The operator $\sqcup$ defines a union of sets assumed to be disjoint.} $U = \sqcup_{D \in \mathcal{D}} \text{interior }D$,
  namely that a generic $v$ restricted to $U$ will also be injective.
  Evidently, $v$ restricted to the complement of $U$ need not be injective,
  since the value of $v$ outside of its support is simply $0$ by definition.  
\end{proof}

Proposition \ref{prop:self_intersection_only_zero} implies that a signal map whose self-intersections are away from $v=0$ is quite unlikely to arise by chance.
This is especially true when large portions of the signal map overlap in $\mathbb{R}^n$, since then the set of self-intersections forms a submanifold of dimension greater than zero.
Intuitively, \emph{if a signal map has substantial self-intersections within a prominent echo, the target in question is likely to be of artificial origin.}

\begin{figure}
  \begin{center}
    \includegraphics[width=4in]{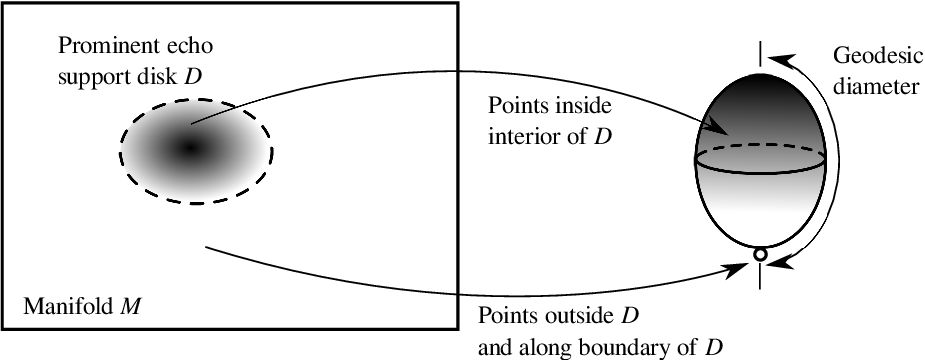}
    \caption{Mapping the boundary of a disk to a point results in a sphere.}
    \label{fig:sphere_forming}
  \end{center}
\end{figure}

\begin{theorem}
  \label{thm:echoes_bouquet_spheres}
  Let $M$ be a smooth manifold and suppose that $n > 2 \dim M$.
  For a generic $v$ in $C_{\mathcal{D}}^\infty(M,\mathbb{R}^n)$,
  the signature space of $v$ is homeomorphic to a wedge product of spheres of (intrinsic) dimension $\dim M$.

  Moreover, each prominent echo corresponds to a sphere of the same dimension as $M$ in the signature space.
  Under the usual metric for $\mathbb{R}^n$,
  the cross section for a prominent echo is a lower bound for the geodesic diameter of its corresponding sphere\footnote{The \emph{geodesic diameter} is the maximum distance between two points, measured via minimum length paths lying on the sphere.}.
\end{theorem}

\begin{proof}
  Under the assumption on dimension, for generic $v$,
  it follows that $v$ restricted to the interior of $D$ is of constant rank and is an immersion.
  According to Proposition \ref{prop:self_intersection_only_zero},
  genericity eliminates self-intersections away from the zero set of $v$.
  In other words, $v$ restricted to the preimage of the set $(\mathbb{R}^n - \{0\})$ is injective.

  Because each closed disk is compact,
  $v$ restricted to the interior of $D$ is an embedding of an open disk as a consequence of \cite[Prop. 7.4]{Lee_2000}.
  
  But $v$ restricted to the boundary of $D$ is identically $0$,
  so the image of $v$ restricted to all of $D$ is homeomorphic a disk quotient its boundary,
  hence a sphere of dimension $\dim M$.
  (See Figure \ref{fig:sphere_forming}.)
  Moreover, all of the restrictions of $v$ to the elements of $\mathcal{D}$ have the same quotient point,
  so the entire image of $v$ is a wedge product of spheres. 
    
  Recall that each disk $D$ is compact.
  Thus there is a point $t \in M$ where $\|v(t)\|_\infty = \sigma_k$ for each each disk $D \in \mathcal{D}$ in the support of $v$.
  On the other hand, $v(p_k) = v(q_k) = 0$ for the endpoints of that same prominent echo.
  According to the triangle inequality in $\mathbb{R}^n$, this means that the geodesic diameter is at least $\ell_k \ge \sigma_k$.
  Briefly, the geodesic diameter comes about by following a path constrained to lie on the surface of the sphere,
  while the distance in $\mathbb{R}^n$ may shortcut through the interior of the sphere.
\end{proof}

The homology of a sphere is very easy to calculate: it has exactly one nontrivial generator in dimension $0$ (the single connected component) and one in the sphere's dimension.
From this observation, one can apply the Mayer-Vietoris exact sequence in the usual way to compute the homology of the signature space.

\begin{corollary}
  \label{cor:homology}
  Under the hypotheses of Theorem \ref{thm:echoes_bouquet_spheres}, the Betti numbers of the signature space are given by
  \begin{equation*}
    \dim H_k(v(M)) = \begin{cases}
      1 & \text{if }k = 0,\\
      \# \mathcal{D} & \text{if } k = \dim M, \\
      0 & \text{otherwise.}
    \end{cases}
  \end{equation*}
\end{corollary}

The converse of Corollary \ref{cor:homology} is immediately useful.
\emph{If the homology of the signature space for a sonar collection disagrees with Corollary \ref{cor:homology}, the target in question is likely to be of artificial origin.}

Consider the case where a monostatic sonar sensor orbits a free-floating object in both azimuth $\theta$ and elevation $\phi$ angles,
which together parameterize the two dimensional sphere $S^2$.
If only the signal strength of the echoes are recorded, this yields a smooth map $u: S^2 \to \mathbb{R}$.
Unfortunately, this means that the hypothesis about the dimension of the codomain in Theorem \ref{thm:echoes_bouquet_spheres} is not satisfied.

This can be remedied somewhat by lifting the signal map $u$ to its \emph{tangent map} $Tu: S^2 \to \mathbb{R}^3$,
which pairs the values of $u$ with the values of its Jacobian matrix.
In the case of this function $u$, the tangent map is given by 
\begin{equation*}
  Tu(\theta,\phi) := \left (u(\theta,\phi), \frac{\partial}{\partial \phi} u(\theta,\phi), \frac{\partial}{\partial \theta} u(\theta,\phi) \right).
\end{equation*}
While the dimension bound is still not satisfied,
the structure as a wedge product of spheres is visible in the image of the tangent map in simple cases as Examples \ref{eg:knotted_sphere_2scatter} and \ref{eg:knotted_sphere_3scatter} below show.

When the image of $Tu$ is not easy to estimate, we can also improve the dimension of the domain using time lagged copies of the signal map.
Using the algorithm described in \cite{Robinson_qplpf},
we may obtain a homeomorphic copy of $C$ from samples of a CSAS collection $u$.
The algorithm requires that we bundle angle-lagged copies of the domain, forming the vector
\begin{equation*}
  \Phi(\theta) := (u(\theta+\tau_1), u(\theta+\tau_2), u(\theta+\tau_3), \dotsc, u(\theta+\tau_N)),
\end{equation*}
where $\tau_1, \dotsc, \tau_N \in S^1$ are constants chosen arbitrarily from a certain dense subset of $(S^1)^N$.
(A classic result from dynamical systems is that most choices for the constants $\tau_1, \dotsc, \tau_N$ are sufficient for the algorithm to work \cite{Takens_1981}.) 
For $N$ large enough, the Jacobian matrix of $\Phi$ becomes nonsingular for all $\theta$.
Therefore, a surjective $\phi$ can be obtained from $\Phi$ by restricting the codomain of $\Phi$ to its image.

\begin{definition}
  \label{def:phase_space}
  If the Jacobian matrix of $\Phi$ has rank equal to the dimension of its domain, we will refer to the image of $\Phi$ as a \emph{phase space} for $u$.
  We will also use the term to refer to the image of $Tu$ in the event that its Jacobian matrix has rank equal to the dimension of its domain.
  Which version we mean will be clear from context.
\end{definition}

\begin{figure}
  \begin{center}
    \includegraphics[width=5in]{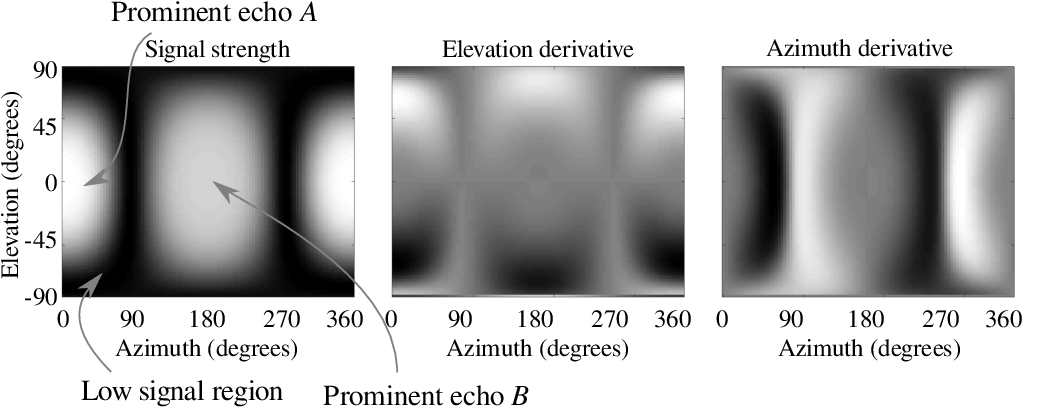}
    \caption{Components of the tangent map of the signal map in Example \ref{eg:knotted_sphere_2scatter}.  Brighter pixels correspond to higher received signal strength.}
    \label{fig:2scatter_data}
  \end{center}
\end{figure}

\begin{figure}
  \begin{center}
    \includegraphics[width=3in]{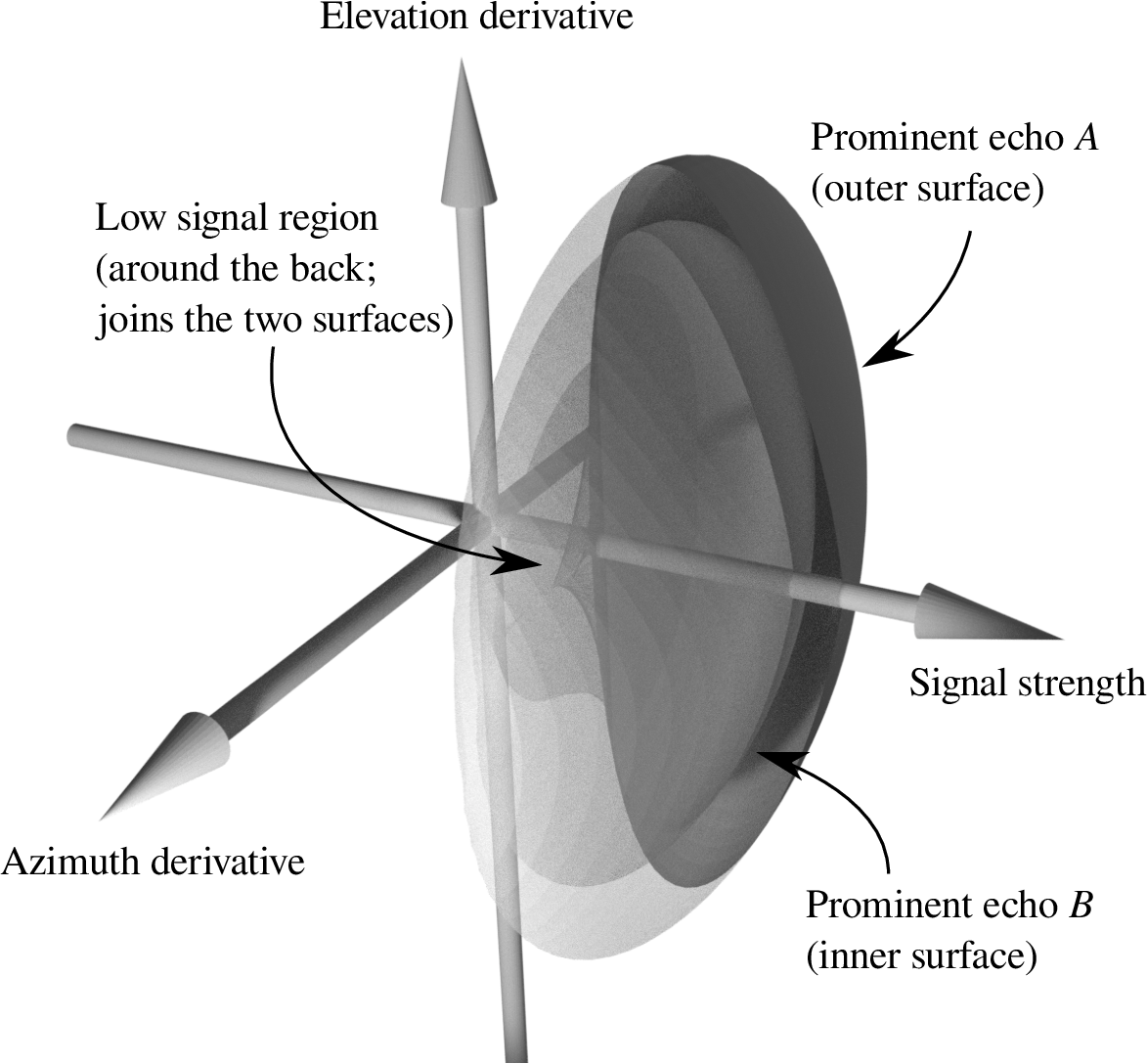}
    \caption{A visualization of the phase space in Example \ref{eg:knotted_sphere_2scatter}, with prominent echoes indicated.}
    \label{fig:2scatter_3d}
  \end{center}
\end{figure}

\begin{example}
  \label{eg:knotted_sphere_2scatter}

  Consider a pair of point scatterers, located in $\mathbb{R}^3$ at $(0,0,0)$ and $(2,0,0)$ (in meters) with complex reflectivity $1$ and $-0.5$, respectively.
  The received signal strength at azimuth $\theta$ and elevation $\phi$ is given by the magnitude of
  \begin{equation*}
    u(\theta,\phi) = \frac{1}{4\pi R}e^{-i k R} - 0.5 \frac{e^{-i k \sqrt{(R\cos\theta\cos\phi - 2)^2 + R^2\sin^2\theta\cos^2\phi+R^2\sin^2\phi}}}{4\pi \sqrt{(R\cos\theta\cos\phi - 2)^2 + R^2\sin^2\theta\cos^2\phi+R^2\sin^2\phi}} ,
  \end{equation*}
  where $R$ is the range from the sensor to the origin, and $k$ is the wavenumber of the transmitted signal.
  
  For $R=4$ meters and $k=\pi/2$, the left frame of Figure \ref{fig:2scatter_data} shows the received signal strength as a function of azimuth and elevation.
  There are two prominent echoes visible, since there are two bright regions (high signal strength) surrounded by a darker region (low signal strength).
  The two prominent echoes therefore each are compactly supported on disks as required by Definition \ref{def:prominent_echo}.
  The cross section of Prominent echo $A$ is slightly larger than that of Prominent echo $B$,
  since the sensor is closer to the scatterers within the support of Prominent echo $A$.
  
  The middle and right frames of Figure \ref{fig:2scatter_data} show the azimuth and elevation partial derivatives, respectively.
  Using the three frames of Figure \ref{fig:2scatter_data} as coordinates for a surface in $\mathbb{R}^3$ yields the phase space for $u$ that is shown in Figure \ref{fig:2scatter_3d}.
  The surface in Figure \ref{fig:2scatter_3d} consists of two nested spheres, joined along a line segment.
  To clearly show this nesting and joining, the surfaces are shown in cross section, with the facing surfaces made transparent.
  
  Because the two prominent echoes differ in cross section, the two nested sphere have somewhat different diameters.
  As a result, they do not self-intersect except along the region outside the support of the prominent echoes.
  This self-intersection is the joining line segment, along the back of the surfaces shown in Figure \ref{fig:2scatter_3d}.
  This is the only self-intersection present, so the phase space as a subset of $\mathbb{R}^3$ has the homotopy type of a wedge product of two $S^2$ spheres.
  In short, even though the dimension bound of Theorem \ref{thm:echoes_bouquet_spheres} is not satisfied,
  its conclusion still holds in this particular (very simple) case.  
\end{example}

When there are more than two scatterers, the failure of the dimension bound in Theorem \ref{thm:echoes_bouquet_spheres} becomes more acute.
Even so, as Example \ref{eg:knotted_sphere_3scatter} shows,
the wedge product structure is still visible after accounting for the presence of self-intersections in the phase space.

\begin{figure}
  \begin{center}
    \includegraphics[width=5in]{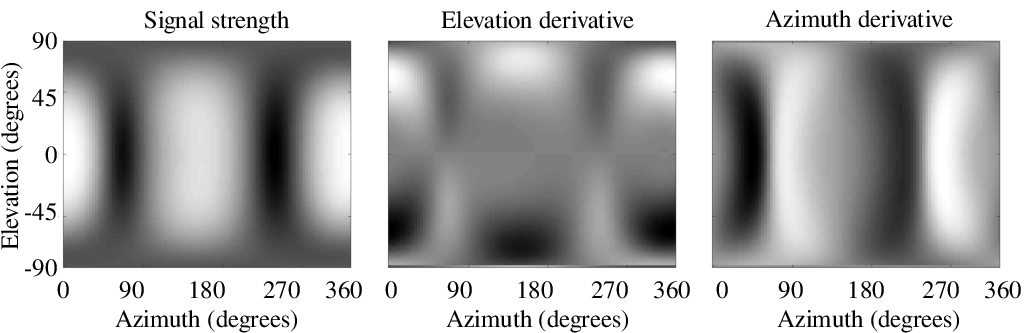}
    \caption{Components of the tangent map of the signature space in Example \ref{eg:knotted_sphere_3scatter}.
      Brighter pixels correspond to higher received signal strength.}
    \label{fig:3scatter_data}
  \end{center}
\end{figure}

\begin{figure}
  \begin{center}
    \includegraphics[width=3in]{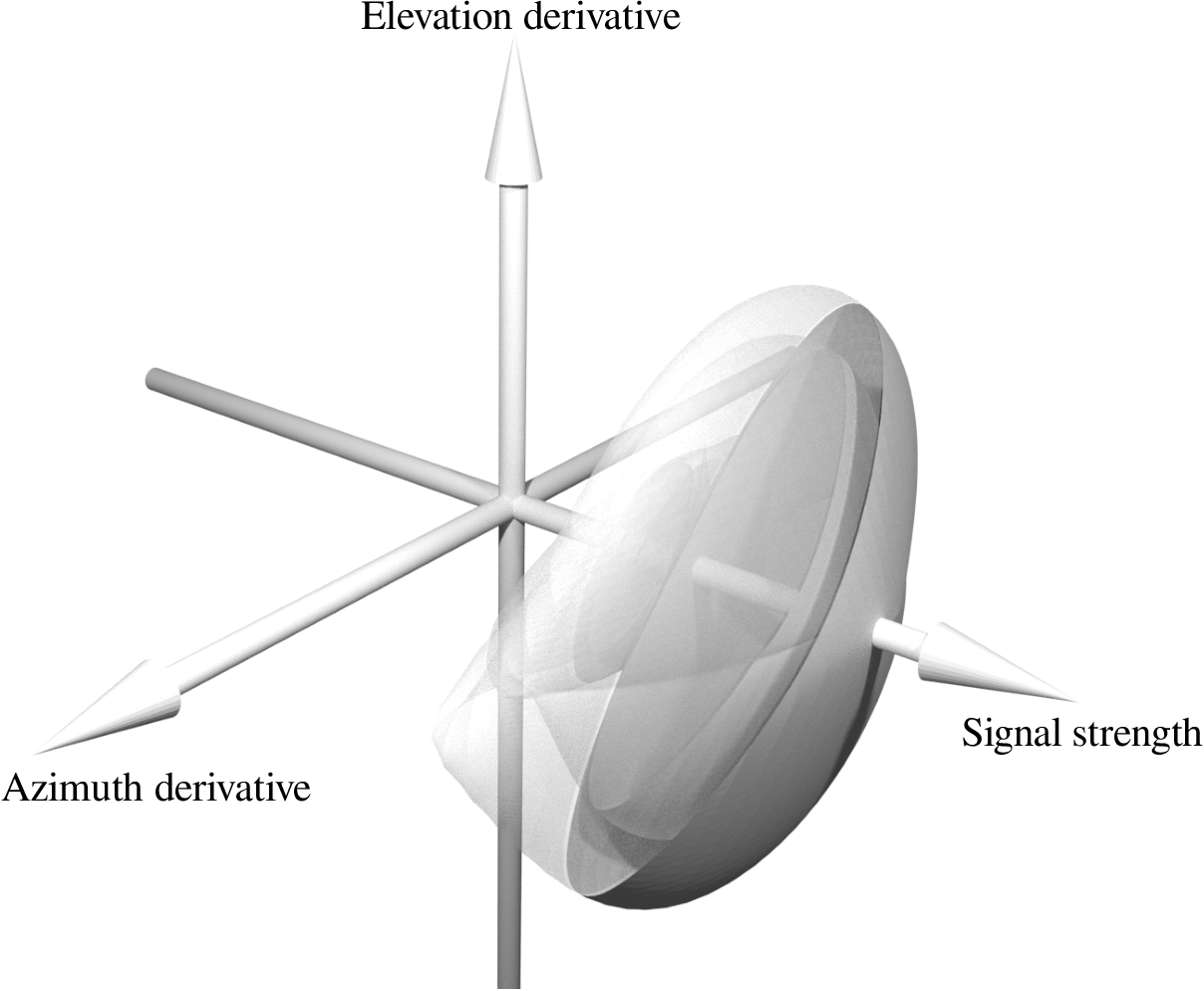}
    \caption{A visualization of the phase space in Example \ref{eg:knotted_sphere_3scatter}.}
    \label{fig:3scatter_3d}
  \end{center}
\end{figure}

\begin{example}
  \label{eg:knotted_sphere_3scatter}
  Consider the same situation as in Example \ref{eg:knotted_sphere_3scatter}, except that instead of two point scatterers, there are three.
  Figure \ref{fig:3scatter_data} shows the resulting signal map for scatterers located at $(0,0,0)$, $(2,0,0)$, and $(0,1,0)$,
  with complex reflectivities $1$, $-0.5$, and $0.75$, respectively.
  (That is, one new scatterer has been added to the situation from Example \ref{eg:knotted_sphere_3scatter}.)
  It is difficult to discern much difference between Figures \ref{fig:2scatter_data} and \ref{fig:3scatter_data},
  though the most apparent thing is that the region of low signal is now no longer a connected region.
  In short, there are three interlocking prominent echoes present.
  It is still the case that the critical points are isolated, however.
  
  Figure \ref{fig:3scatter_3d} shows the resulting phase space.
  There are still the same two prominent echoes visible as nested spheres as before,
  along with the seam along the back.
  However, there is now a structure present behind the original two spheres.
  Because the dimension hypothesis of Theorem \ref{thm:echoes_bouquet_spheres} is not satisfied,
  there can be---and now are---spurious self-intersections of the surface that make interpretation quite difficult.
  Close examination of the surfaces show that the \emph{single} new sphere guaranteed by Theorem \ref{thm:echoes_bouquet_spheres}
  \emph{when the dimension bound is satisfied} is present in one of the new spheres attached on the back.
  The others are actually not separate spheres, but rather result from a self-intersecting twist applied to the spheres on the front of the plot.
\end{example}

\subsection{Special results for circular SAS}
\label{sec:embedded_csas}

When a sonar system orbits a target along a closed path,
the domain of the signal map will be $M=S^1$, a topological circle.
If instead, the sonar system simply moves past a target along a non-closed arc in space, 
the domain $M$, namely the set of locations of the sonar system, will be homeomorphic to an open interval of $\mathbb{R}$.
These are the only possibilities when $M$ is a $1$-dimensional manifold without boundary:
it is diffeomorphic to either $S^1$ or $\mathbb{R}$.

In either case,
let us assume that a signal map $v$ has a nonvanishing derivative on the interior of its support.
This means that $v$ is of constant rank (which is $1$) on the interior of its support.
(This is the generic case as per \cite[Thm. 10.8]{Lee_2000} provided $n \ge 2$, so it is not onerous to assume in practice.)

The topology of a $1$-dimensional manifold $M$ is generated by intervals.
Therefore, $C_{\mathcal{D}}^\infty(M,N)$ is the collection of smooth functions whose support is contained in a disjoint union of intervals.
For a given $v$, this means that 
\begin{equation*}
  \supp v = \cl \{ t \in M : v(t) \not=0 \} = \sqcup \mathcal{D} = \bigsqcup_k [p_k,q_k],
\end{equation*}
where we will permit $p_k$ to take the value $-\infty$ or $q_k$ to take the value $+\infty$.
(Notice that this can only happen at most once for each endpoint.)
According to Definition \ref{def:prominent_echo},
each interval in the union above is a prominent echo with cross section $\sigma_k = \|v(t) 1_{[p_k,q_k]}(t)\|_\infty$.

The lack of self-intersections implied by Theorem \ref{thm:echoes_bouquet_spheres} for generic signal maps means that when the domain is $1$-dimensional,
as it is for CSAS,
the resulting signature space is an embedded graph.

\begin{corollary}
  \label{cor:echoes_embedded_graph}
  If $M$ is a $1$-dimensional manifold and $n>2$,
  then $v : M\to \mathbb{R}^n$ is generically an embedding of a graph upon descending to the quotient $M/(M \backslash \supp v)$.
\end{corollary}

Rather than requiring genericity, a signal map with isolated critical points is sufficient to establish a result like Theorem \ref{thm:echoes_bouquet_spheres}.
Recall that within a prominent echo,
a generic signal map will have isolated critical points.

\begin{proposition}
  \label{prop:csas_direct}
  Suppose that $u: S^1 \to \mathbb{C}^n$ is a smooth function with isolated critical points.
  Then the signature space of $u$ is homeomorphic to a path-connected, compact, $1$-dimensional cell complex.
  Moreover the signature space is homotopy equivalent to a wedge product of finitely many circles.
\end{proposition}
\begin{proof}
  Since $u$ is smooth, it is continuous.
  Therefore, the signature space, the image of $u$, inherits path-connectedness, compactness, and its maximum dimension (namely dimension $1$) from $S^1$.

  Every path-connected, $1$-dimensional cell complex with finitely many cells is homotopy equivalent to a wedge product of circles.
  To see this homotopy equivalence, recognize that such a cell complex is homotopy equivalent to the geometric realization of a finite undirected graph that is itself path connected.
  Such a graph has a spanning tree via many well-known algorithmic constructions.
  Collapsing the interior of this spanning tree to a point provides the deformation retraction (a homotopy equivalence) to a wedge product of circles.

  We now need to establish that there are finitely many cells, which we do by construction.
  Because the critical points are assumed to be isolated, compactness implies that there are finitely many critical points.
  We may therefore partition $S^1$ into the disjoint union of these critical points and finitely many connected open intervals.
  This yields a cell complex structure for $S^1$.
  The inverse function theorem ensures that $u$ is a local diffeomorphism (hence local homeomorphism) on each of the open intervals.
  The image of such an open interval is therefore also an interval.
  We therefore have a partition of the image of $u$ into a disjoint union of finitely many $0$-dimensional cells (the critical points) and finitely many $1$-dimensional cells (the intervals).

  It merely remains to show that each interval is attached to a critical point at each of its ends.
  Consider an interval that is the image of $(\theta_1,\theta_2) \subseteq S^1$.
  By construction, $\theta_1$ and $\theta_2$ are critical points of $u$, which may be the same point.
  Therefore, $u((\theta_1,\theta_2))$ is attached to these two critical points.
  In short, considering the cell complex decomposition of $S^1$, every critical point corresponds to a vertex of degree $2$. 
\end{proof}

The signature space of $u$ may not be \emph{homeomorphic} to a wedge product of finitely many circles.
The cell complex structure for the signature space might contain degree $1$ vertices,
which can occur when the intervals on either side of such a vertex map to the same locus of points.

\begin{figure}
  \begin{center}
    \includegraphics[width=4in]{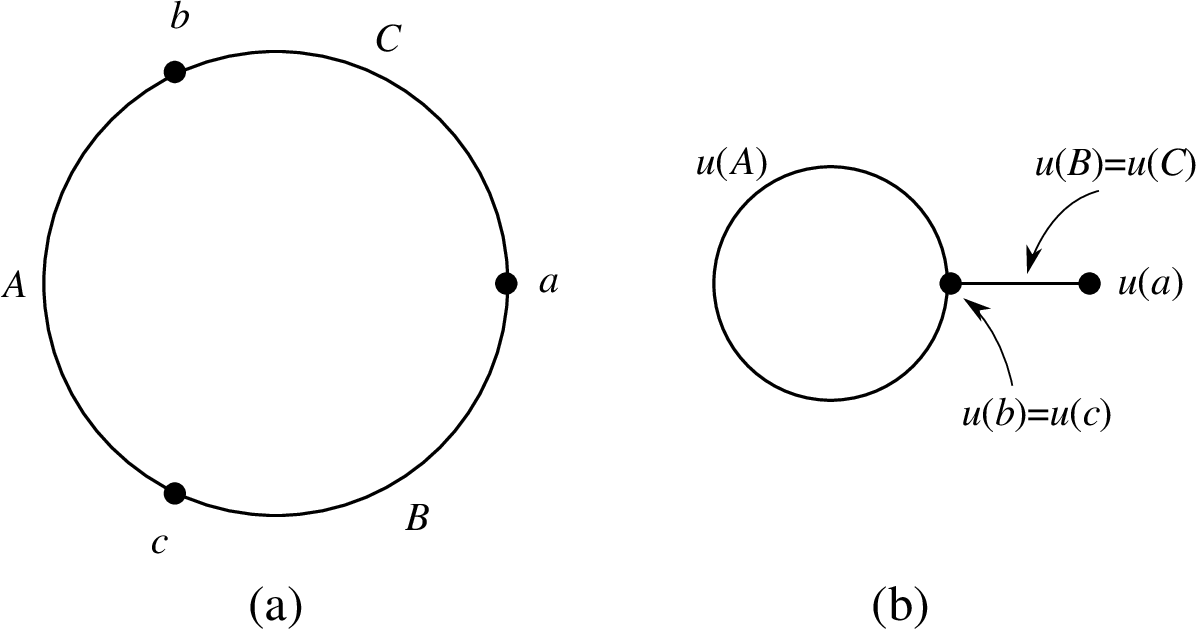}
    \caption{Formation of a flare in the image of a smooth map, as described in Example \ref{eg:collapse}.
      (a) A cell complex decomposition of the circle.
      (b) The image of the smooth map $u$, in which two of the $1$-dimensional cells are identified.}
    \label{fig:collapse}
  \end{center}
\end{figure}

\begin{example}
  \label{eg:collapse}
  Consider the case of a smooth map $u : S^1 \to \mathbb{R}^2$ with three critical points $a,b,c \in S^1$.
  This induces a cell complex decomposition of the circle as shown in Figure \ref{fig:collapse}(a).
  The map may be defined so that the critical point $a$ is flanked by two intervals with the same image as is shown in Figure \ref{fig:collapse}(b).
  Notice that the degree of $a$ is $2$, but that the degree of $u(a)$ is $1$.
\end{example}

A $1$-dimensional cell with a degree $1$ endpoint is called a \emph{flare}.
In CSAS data, flares correspond to specular reflections that are left/right symmetric.
Since flares constitute regions where the signature space map is not injective,
it should be noted that flares are usually \emph{not} seen in generic signal maps due to Proposition \ref{prop:self_intersection_only_zero}.
Natural objects in a scene typically do not exhibit strong symmetries,
and as a result are well described by a generic signal map.
\emph{The presence of flares in CSAS data is usually a clear indication that the scene contains objects of artificial origin.}

On the other hand, generic functions are not symmetric, and moreover generic smooth functions have Jacobians of constant rank.
Restricting to the image of the signal map yields a \emph{surjective submersion}.

\begin{definition}
  A \emph{submersion} is a smooth function $f: M \to N$ between two manifolds,
  such that the Jacobian matrix at every point of $M$ is surjective.
\end{definition}

Submersions are constant rank maps, and therefore cannot have critical points \cite[Thm. 7.9]{Lee_2000}.

A smooth quotient is induced by a surjective submersion, which ensures that the codomain is topologically similar to the domain. 
The surjective submersion assumption is a powerful inferential tool in the case of CSAS,
because the domain $M$ has a very constrained and well-understood structure, namely it is a circle.
Recall in particular that the phase space (Definition \ref{def:phase_space}), having a Jacobian of maximal rank, is a surjective submersion onto its image.

\begin{proposition}
  \label{prop:phase_structure}
  Suppose that $\phi: S^1 \to C$ is a surjective submersion of the circle.
  Then $C$ is homeomorphic to a path-connected, compact, $1$-dimensional cell complex in which no vertex has degree $1$.
\end{proposition}

\begin{proof}
  Since $\phi$ is smooth, it is continuous.  $C$ inherits path-connectedness, compactness, and its dimension from $S^1$, because $\phi$ is also surjective.

  It remains to demonstrate degree $1$ vertices contradict the assumptions on $\phi$.
  Suppose that $C$ does have a degree $1$ vertex.
  This means that there is a subspace of $C$ that is homeomorphic to $[0,1)$,
    where the point $0$ maps to the degree $1$ vertex.
    Notice that since the domain is $1$-dimensional,
    the Jacobian matrix of $\phi$ is just the usual derivative $\phi'$.
    This derivative must evidently must take the value $\phi'(0)=0$ due to the fact that there are no points ``left of'' $0$ in $C=[0,1)$.
      This means that $\phi$ has a critical point at $0$, which contradicts the fact that $\phi$ is a submersion and therefore its rank cannot take the value $0$.
\end{proof}

In short, because of the definition of the phase space (Definition \ref{def:phase_space}), Proposition \ref{prop:phase_structure} asserts that the phase space cannot contain flares.
In practice, one might expect to find apparent occurrences of flares in sampled data, due to the presence of aliasing and sampling error.

\subsection{Superposition of point scatterers}
\label{sec:point_scatterers}

The signal received at a point $y \in \mathbb{R}^d$ from a superposition of $p$ point scatterers at locations $x_1, \dotsc, x_p \in \mathbb{R}^d$ with complex reflectivities $a_1, \dotsc, a_p \in \mathbb{C}$ collected at $\ell$ fixed wavenumbers $k_1 < \dotsb < k_\ell$  is a function $v : \mathbb{R}^d \to \mathbb{C}^\ell$ of the form
\begin{equation}
  \label{eq:point_scatterers}
  v_q(y) = \sum_{j=1}^p a_j g_{k_q}(\|x_j - y\|), \text{ for each } q = 1, \dotsc, \ell,
\end{equation}
where
\begin{equation}
  \label{eq:green}
  g_k(r) = \frac{1}{r} e^{-i k r}
\end{equation}
is the Green's function for the wave operator.

\begin{lemma}
  \label{lem:green_injective}
  If $k >0$, the function $g_k : (0,\infty) \to \mathbb{C}$ defined in Equation \eqref{eq:green} is injective.
\end{lemma}
\begin{proof}
  If $k>0$, the magnitude $|g_k(r)|$ is simply $1/r$, which is an injective function for $r > 0$.
\end{proof}

The following Lemma \ref{lem:distances_general_position} is standard.

\begin{lemma}
  \label{lem:distances_general_position}
  The continuous function $m : \mathbb{R}^d \to \mathbb{R}^{d+1}$ given by
  \begin{equation*}
    m(y) := (\|x_0 - y\|, \dotsc, \|x_d - y\|),
  \end{equation*}
  the vector of distances to $d+1$ points, is injective for a generic choice of points $\{x_0, \dotsc, x_d\}$.
\end{lemma}

Signal maps formed by a linear superposition of point scatterers are not generic in the space of smooth functions.
However, of the signal maps arising as linear superpositions of point scatterers, a generic subset \emph{of these} are also injective.
This is the distinction between Proposition \ref{prop:point_injectivity} below and Theorem \ref{thm:echoes_bouquet_spheres} proven earlier.

\begin{proposition}
  \label{prop:point_injectivity}
  Consider the signal received from a superposition of $p$ point scatterers collected at $\ell$ fixed wavenumbers $k_1 < \dotsb < k_\ell$ in $d$ spatial dimensions,
  given by the function $v: \mathbb{R}^d \to \mathbb{C}^\ell$ defined by Equation \ref{eq:point_scatterers}.
  If $2\min(p,\ell) > d$, the function $v$ is injective for generic values of $a_j$ and $x_j$ for $j=1,\dotsc,p$ and generic wavenumbers $k_q$ for $q=1,\dotsc,\ell$.
\end{proposition}
\begin{proof}
  For notational convenience, let
  \begin{equation*}
    r_j := \|x_j - y\|.
  \end{equation*}
  We can organize the information in $v$ as a matrix equation
  \begin{equation*}
    \begin{aligned}
      v(y) &= \begin{pmatrix}a_1 & \dotsb & a_p \end{pmatrix}
      \begin{pmatrix} g_{k_1}(r_1) & \dotsb & g_{k_\ell}(r_1) \\ \vdots & & \vdots \\ g_{k_1}(r_p) & \dotsb & g_{k_\ell}(r_p) \end{pmatrix} \\
      &=\begin{pmatrix}a_1 & \dotsb & a_p \end{pmatrix}
      \begin{pmatrix} \frac{1}{r_1}e^{i k_1 r_1} & \dotsb &  \frac{1}{r_1}e^{i k_\ell r_1} \\ \vdots & & \vdots \\  \frac{1}{r_p}e^{i k_1 r_p} & \dotsb &  \frac{1}{r_p}e^{i k_\ell r_p} \end{pmatrix} \\
    \end{aligned}
  \end{equation*}
  Note that the output of $v$ is a $1 \times \ell$ row vector.
  Although the components of the matrices on the right depend nonlinearly on $y$,
  that is the only nonlinearity present in $v$.
  According to Lemmas \ref{lem:green_injective} and \ref{lem:distances_general_position},
  for generic choice of $x_j$ for $j=1,\dotsc,p$,
  the function $\mathbb{R}^d \to \mathbb{C}^p$ given by applying $g_{k_q}$ to each $r_j$ is injective in $y$ for each $q = 1, \dotsc, \ell$.
  This means that each of the columns in the matrices on the right above are injective functions of $y \in \mathbb{R}^d$.
  A full rank square matrix can be factored out as follows,
  \begin{equation*}
    \begin{aligned}
      v(y) &= \begin{pmatrix}a_1 & \dotsb & a_p \end{pmatrix}
      \begin{pmatrix}\frac{1}{r_1} & \dotsb & 0 \\ \vdots & & \vdots \\ 0 & \dotsb & \frac{1}{r_p}\end{pmatrix}
      \begin{pmatrix} e^{i k_1 r_1} & \dotsb &  e^{i k_\ell r_1} \\ \vdots & & \vdots \\  e^{i k_1 r_p} & \dotsb & e^{i k_\ell r_p} \end{pmatrix}.
    \end{aligned}
  \end{equation*}
  Careful inspection reveals that the remaining matrix on the right is a Vandermonde matrix,
  and is therefore of full real rank, namely $2\min(p,\ell)$, under generic choice of wavenumbers $k_q$.
  Furthermore, under generic choice of $a_j$ for $j=1,\dotsc,p$, the matrix on the left will be of full real rank $2$.
  This means that if $2\min(p,\ell) > d$ the function $v$ will be injective under the stated genericity assumptions.
\end{proof}

Proposition \ref{prop:point_injectivity} implies that for sufficiently many point scatterers and collected wavenumbers,
the received signal is typically supported on all of the domain.
Neither Example \ref{eg:knotted_sphere_2scatter} nor Example \ref{eg:knotted_sphere_3scatter} satisfy the hypotheses of Proposition \ref{prop:point_injectivity} because in the $d=3$ in both cases, and the signature space arises from $2$ and $3$ point scatterers, respectively.
This explains the presence of prominent echoes in these two examples.

However, once there are sufficiently many point scatterers, Proposition \ref{prop:point_injectivity} implies that self-intersections in the signature space, and hence prominent echoes, vanish.

\begin{corollary}
  \label{cor:distributed}
The signal map from a generic superposition of a sufficiently many point scatterers is not an element of
$C_{\mathcal{D}}(S^{d-1},\mathbb{C}^n)$ unless $\mathcal{D} = \{S^{d-1}\}$, and this cannot occur because $S^{d-1}$ is not a disk.
Put another way, if the signal map has nontrivial prominent echoes,
the signal arises from distributed---not point---scatterers.
\end{corollary}

\subsection{Persistent homology}

\label{sec:ph}

Proposition \ref{prop:point_injectivity} implies that for sets of point scatterers,
the signature space usually has the same topology as the domain of the signal map.
As a result, all purely topological information in the signature is simply a recapitulation of what was already present.
On the other hand, for distributed scatterers, Corollaries \ref{cor:homology} and \ref{cor:distributed} indicate that there can be topological information in the signature space.
Therefore, the signature space has both topological and geometric information that was not present in the domain.

A convenient tool set for extracting both geometric and topological information is \emph{persistent homology} \cite{bobrowski2023universal,bobrowski2022universality}.
Although originally developed to compute topological features from noisy information,
it is now becoming clear that persistent homology is extremely useful for detecting geometric features that enable high classification accuracy \cite{Adams2021TopologyAT,Bubenik2019PersistentHD}.
For instance, persistent homology can detect the number of loops (prominent echos), curvature (related to beamwidth), and local convexity of the signature \cite{turkes2022effectiveness}.
Although we encourage the reader to consult one of the many standard references on persistent homology (for instance \cite{Carlsson_Vejdemo-Johansson_2021}),
we give a brief summary here.

Instead of considering the signature space $v(M)$ directly, persistent homology has us consider a \emph{filtration} of \emph{abstract simplicial complexes} constructed from the signature space.
An \emph{abstract simplicial complex} consists of a set of vertices and a collection of subsets of vertices, called \emph{simplices},
such that all of the subsets of a given simplex are also simplices.
Associated to each abstract simplicial complex is a set of \emph{homology generators}, each of which represents a void of a particular dimension.
Because CSAS signatures are $1$-dimensional cell complexes by Corollary \ref{cor:echoes_embedded_graph},
of interest in this article are the voids of dimension $1$, which correspond to non-collapsable loops in the abstract simplicial complex.

We will consider the \emph{Vietoris-Rips filtration} $VR_\epsilon$ of the signature space, in which each point of the signature space $x \in v(M)$ will be deemed a vertex.
Which simplices are present are determined by the \emph{filtration parameter} $\epsilon > 0$.
If $x_0$, $\dotsc$, $x_k$ are points in $v(M)$, and all pairwise distances between these points are less than $\epsilon$, then $[x_0, \dotsc, x_k]$ is a simplex in $VR_\epsilon$.

With this construction in hand, a \emph{persistence diagram} is a scatter plot in the plane where each point shown is a generator for $VR_\epsilon$.
The coordinates of each point represent the filtration parameters over which the generator  is present.
The horizontal axis shows its \emph{birth}, the persistence parameter it comes into existence, while the vertical
axis shows its \emph{death}, the persistence parameter at which void is no longer present.
The difference between birth and death is usually called \emph{persistence}, and is a measure of the corresponding generator's robustness to perturbations.
Geometrically, persistence measures the distance between a point and the diagonal in the persistence diagram.

\begin{proposition}
  \label{prop:sonar_pd}
  For the $1$-dimensional embedded Vietoris-Rips filtration $VR_\epsilon$ constructed from the signature space as above,
  each prominent echo corresponds to a generator with death time bounded below by $\sigma/2$.
\end{proposition}
The reader can also consider \cite[Thm. 4]{Rieck_2023}, wherein persistent homology in dimensions $0$ and $1$ is sufficient to bound the length of the shortest cycle in a graph.
\begin{proof}
  Suppose that $M$ is a $1$-dimensional manifold without boundary and that $v \in C_{\mathcal{D}}(M,\mathbb{R}^n)$.
  According to Theorem \ref{thm:echoes_bouquet_spheres} and Corollary \ref{cor:echoes_embedded_graph},
  for sufficiently large $n$,
  the signature space $v(M) \subseteq N$ has a sphere of dimension $1$, namely a loop, for each prominent echo of $v$,
  whose geodesic diameter is bounded below by the cross section $\sigma$ of the echo.
  (In agreement with this fact, Corollary \ref{cor:homology} indicates that the $1$-Betti number of $v(M)$ is the number of prominent echos.)

  This means that the signature space has the structure posited in \cite[Thm 1.1]{Gasparovic_2018},
  namely it is a metric graph in which the length of each loop is bounded below by twice its diameter.
  Given this information, \cite[Thm 1.1]{Gasparovic_2018} implies that the $1$-dimensional generators for the persistence diagram using the geodesic metric can be computed exactly,
  \begin{equation*}
    \left\{\left(0,\frac{\ell_k}{4}\right) : 1 \le i \le k\right\},
  \end{equation*}
  in which $\ell_k$ is the length of the $k$-th loop, the first coordinate corresponds to birth, and the second coordinate corresponds to death.
  Since $\ell_k \ge 2 \sigma_k$, where $\sigma_k$ is the cross section of the $k$-th prominent echo, we have that each prominent echo corresponds to a death time of at least $\sigma_k/2$.

  While useful, \cite[Thm 1.1]{Gasparovic_2018} does not completely finish the argument because we are actually interested in the Vietoris-Rips filtration in $\mathbb{R}^n$ with the Euclidean metric.
  On the other hand, \cite[Prop 2.6]{Carlsson_2009} establishes that the lower bounds on the death times are the same for $VR_\epsilon$ as for the geodesic metric in this case.
\end{proof}

Note that Proposition \ref{prop:sonar_pd} is for an infinite sampling rate.
For a finite sample rate, the birth for the generators in the persistence diagram is no longer always $0$,
but is is related to the sample rate,
the trajectory's velocity,
and the beamwidth of the scatterer associated to the loop.

\section{Experimental verification}
\label{sec:verification}

In this section, we show that in simulation and in laboratory conditions, 
the signature space for CSAS consists of a wedge product of loops, one for each prominent echo, whose diameters are governed by sonar cross section (thus verifying Theorem \ref{thm:echoes_bouquet_spheres} and Proposition \ref{prop:csas_direct}).
We also establish that the bound on the death persistence parameter for each loop is bounded below by half its sonar cross section, thus verifying Proposition \ref{prop:sonar_pd}.
In simulation, where we can vary noise levels arbitrarily, we also establish that the bound on the death persistence parameter is robust to perturbation.

Being primarily concerned with the theoretical conditions that underly successful classification using persistent homology,
the present article does not attempt to perform a systematic study of sonar classification performance.
We take several illustrative examples of laboratory and simulated data to demonstrate the validity of our (generally rather weak) assumptions about the structure of the data.
We also aim for clear interpretation,
and therefore do not attempt to use automated tools for matching barcodes \cite{redondo_2024}.

Since Theorem \ref{thm:echoes_bouquet_spheres} establishes a particular cell complex structure for the space of sonar echos,
we aim to show that this is visible in the data.
It is not the case that cell complex structure can be inferred completely from persistent homology,
but for limited cases it can rule out certain structures.
However, because we show that the cell complex structure is of sufficiently low dimension,
a visual inspection of the inferred spaces is actually more informative.

\subsection{Simulation}

\label{sec:sim_expt}

In order to provide a simple test of the theoretical results,
we begin with a simulation of a CSAS collection of echos from a short segment of pipe.
This simulation is constructed by uniformly placing point scatterers of equal reflectivity along the sides of a narrow rectangle.

\begin{figure}
  \begin{center}
    \includegraphics[width=5.5in]{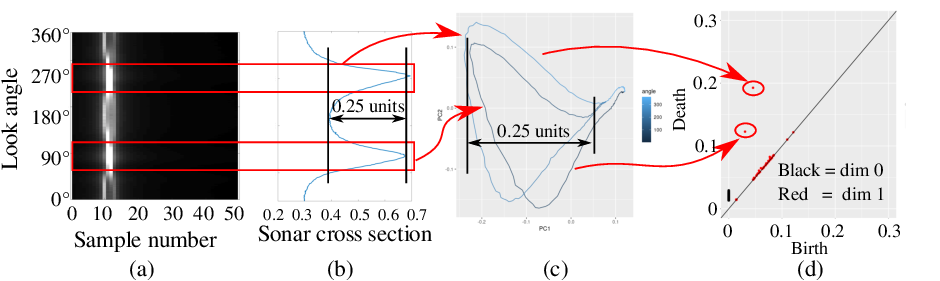}
    \caption{Simulation data for a CSAS image of a short segment of pipe: (a) raw echo data in which brightness of a pixel corresponds to relative signal strength, (b) sonar cross section as a function of look angle, (c) phase space, (d) persistence diagram.  Note that the vertical scales on (a) and (b) are the same, namely look angle in degrees.  Two prominent echos are marked in (a) and (b) by rectangles.}
    \label{fig:pipe_open_sim}
  \end{center}
\end{figure}

As expected, the lateral faces of the pipe segment result in strong specular reflections at $90^\circ$ and $270^\circ$,
as is shown in Figure \ref{fig:pipe_open_sim}(a) and (b).
Although not perfectly isolated from one another---the simulation is the solution of a linear Helmholtz equation---these correspond to prominent echos in the theoretical model.
These prominent echos correspond to two loops in both the signature space and the phase space.
These loops are visible in phase space shown in Figure \ref{fig:pipe_open_sim}(c).

The distance of a point from the diagonal in a persistence diagram, such as what is shown in Figure \ref{fig:pipe_open_sim}(d), is the \emph{lifetime} of the topological feature.
Notice that the lifetime is precisely the difference between death and birth times.
Longer lifetimes are associated with more robust features.
In Figure \ref{fig:pipe_open_sim}(d),
there is one clear long lifetime feature for $H_1$ and one shorter (yet still long) lifetime feature.
These two features correspond to the specular reflections from the lateral faces.
This is a clear confirmation of Proposition \ref{prop:phase_structure}.

Due to the geometric symmetry of the point scatterers,
which means that the genericity hypothesis of Proposition \ref{prop:point_injectivity} is not satisfied,
the signature space map is not injective.
The two loops in Figure \ref{fig:pipe_open_sim}(c) nearly overlap as a result.

However, aside from that symmetry, the loops do not otherwise self-intersect in the signature space,
except near the origin as anticipated by Theorem \ref{thm:echoes_bouquet_spheres},
even though this is a little hard to see in Figure \ref{fig:pipe_open_sim}(c).

Theorem \ref{thm:echoes_bouquet_spheres} claims that the maximum diameter of the loops will be not more than the maximum sonar cross section within a prominent echo.
For each loop this is $0.7$ units, well above the observed diameter of both loops, which is approximately $0.25$ units for both loops  (seen from Figure \ref{fig:pipe_open_sim}(c)).
Although consistent with Theorem \ref{thm:echoes_bouquet_spheres}, the disparity between the bound and the actual value is mostly accounted for the fact that the sonar cross section does not drop to zero at $0^\circ$ and $360^\circ$.
The difference in sonar cross section between the specular reflection and the null is about $0.25$ units.
Using the value of $\sigma = 0.25$ to account for the fact that the signal does not fully drop to zero outside the prominent echos,
then Proposition \ref{prop:sonar_pd} asserts that the death time of the loops should be bounded below by $0.12$ units,
and this is confirmed in the persistence diagram in Figure \ref{fig:pipe_open_sim}(d).

\begin{figure}
  \begin{center}
    \includegraphics[width=4in]{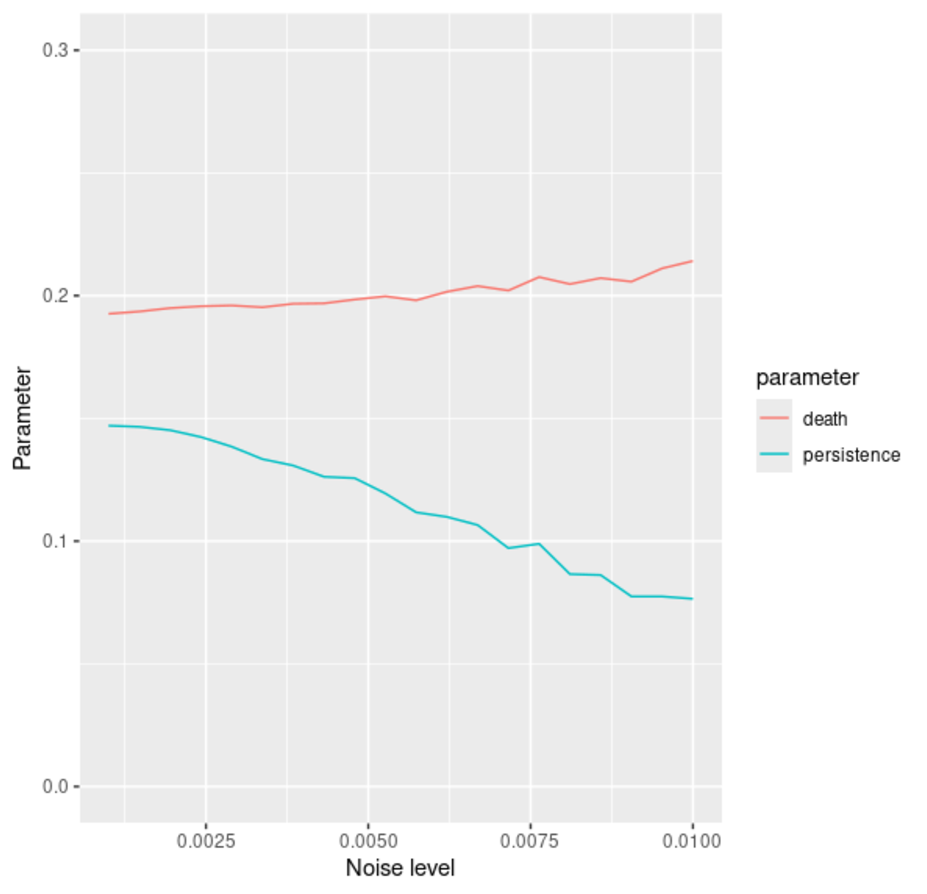}
    \caption{Value of the death time and persistence value for the most persistent loop for a CSAS image of a short segment of pipe under varying noise levels.}
    \label{fig:pipe_open_noise}
  \end{center}
\end{figure}

To establish that these findings are robust under noise, we applied additive white Gaussian noise to the simulated receiver at varying levels.
Figure \ref{fig:pipe_open_noise} shows that as the noise level increases, the death time of the most persistent loop remains stable in accordance with Proposition \ref{prop:sonar_pd}.
On the other hand, the persistence value of this loop decreases, which indicates that the topological structure of the signature is degrading as expected.
The horizontal axis in Figure \ref{fig:pipe_open_noise} uses the same units as that of sonar cross section, but is stated per range cell.
As a result, the right side of Figure \ref{fig:pipe_open_noise} corresponds to a signal-to-noise ratio of $10 \log_{10} (0.7 / (0.01 \sqrt{100})) = 8$ dB.

\subsection{Laboratory dataset description}
\label{sec:expt_data}

 We analyzed CSAS collections of three household objects (a copper pipe, a styrofoam cup, and a coke bottle) in two configurations
 (open ends or capped ends) each using ARL/PSU's AirSAS system \cite{cowenairsas}.

 The objects were placed at the center of a rotating turntable and were imaged from a fixed sonar sensor located a few meters away from the center.
Each of the three objects are (mostly) rotationally symmetric around a single axis.
This axis was oriented approximately perpendicularly with respect to the axis of rotation of the turntable,
so that the objects were ``laid on their side'' as they were rotated.

The collected data consist of one two dimensional array of data for each configuration,
in which the rows correspond to the turntable rotation angle (briefly, the \emph{look angle}) and the columns correspond to range.
The array for each target was the same shape, and contains one sample per degree (a total of $360$ rows) and $1000$ range samples (columns).
The raw sonar data are shown in the (a) frames of Figures \ref{fig:pipe_open_triptych}---\ref{fig:coke_capped_triptych}.


\begin{figure}
  \begin{center}
    \includegraphics[width=5in]{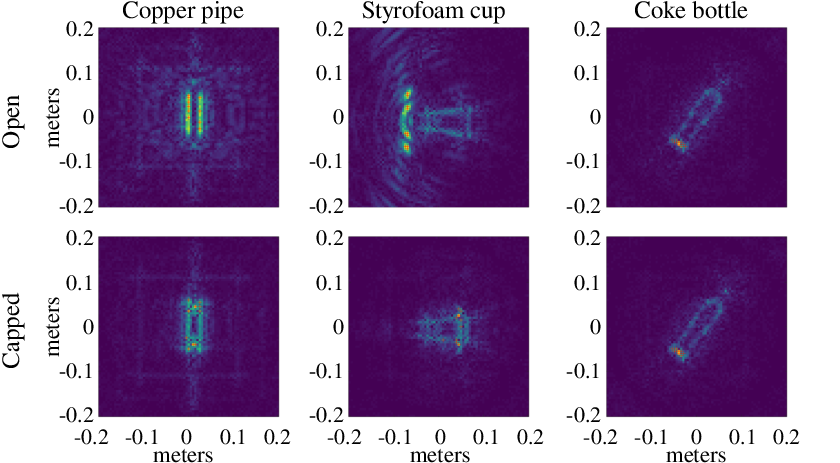}
    \caption{Time domain backprojection CSAS images of the three targets (in both capped and open configurations) used in the experiment. The turntable rotation axis is perpendicular to the page, located at the origin.  Color intensity indicates received sonar cross section.}
    \label{fig:airsas_targets}
  \end{center}
\end{figure}

The three targets have sonar cross sections that depend on look angle in distinctly different ways,
as shown in Figure \ref{fig:airsas_targets}.
Although the methods discussed in this paper use the raw sonar echos, Figure \ref{fig:airsas_targets} shows the time domain backprojection images of each of the six configurations collected to give the reader an idea of both the shape of the objects and the physical acoustical processes represented in the data.

One thing is immediately clear from the raw data: there are isolated prominent echos.
As a result, Corollary \ref{cor:distributed} asserts that each object contains distributed scatterers, a fact which is easily verified in Figure \ref{fig:airsas_targets}.
The copper pipe and styrofoam cup show strong specular reflections with narrow beam width from their lateral faces,
because these faces are large flat surfaces.
The lateral faces of the coke bottle also shows a strong specular reflection,
though its beamwidth is much larger due to curvature.
When the ends of each object are closed,
the faces presented are convex and much smaller than a wavelength.
As a result, when the objects are oriented so that the closed ends face the incident wave,
they do not have a large sonar cross section.
When the ends are open,
the objects can resonate and reradiate the sound energy,
so large non-specular reflections from the open ends are possible.
Since the reradiation process takes time,
there is dispersion in time and range on these echoes.

\subsection{Laboratory results}

\label{sec:expt}

This section shows that the phase space of each object collected by the AirSAS system is a wedge product of circles, in which each circle corresponds directly to a prominent echo, whose diameter bounds the maximum sonar cross section in the prominent echo, 
as guaranteed by Theorem \ref{thm:echoes_bouquet_spheres}.

 For each target in each configuration, we estimated the phase space (Definition \ref{def:phase_space}) using a sliding window with $N=3$ angle-lagged offsets.
 Since many choices of the offsets will work,
 we arbitrarily chose to use $\tau_1 = 0^\circ$, $\tau_2 = 4^\circ$, and $\tau_3 = 25^\circ$ offsets for the results shown here.

 We can visualize the phase space using \emph{principal components analysis} (PCA).
 The six experimental configurations are ordered in increasing topological complexity below,
 as the loops shown in the PCA coordinates become progressively harder to identify.
 In each case, the loops can be identified by comparing their look angle at maximum distance from the origin against the look angle of the prominent echos visible in the raw data.
 
 \paragraph{Copper pipe with open ends}

\begin{figure}
  \begin{center}
    \includegraphics[width=5in]{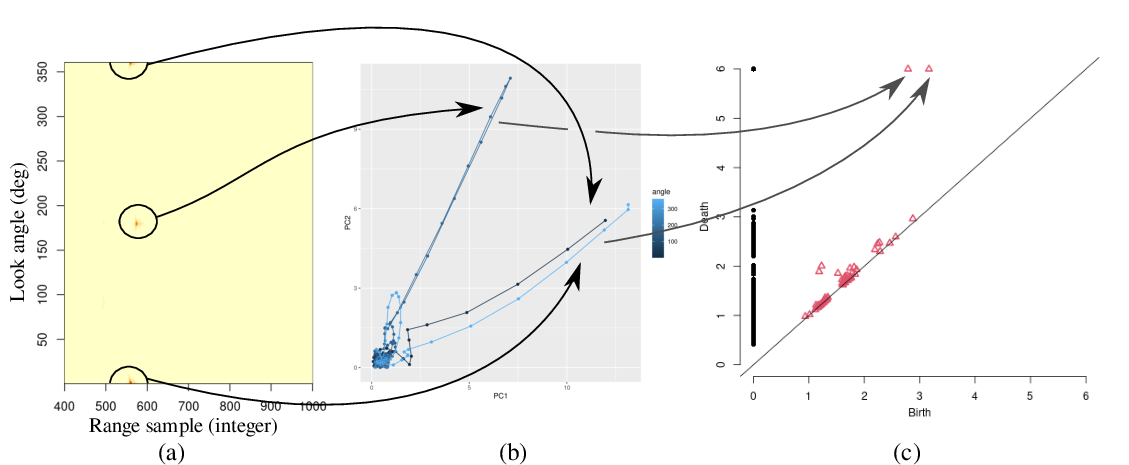}
    \caption{The copper pipe with open ends: (a) raw AirSAS data with prominent echos marked, (b) phase space shown in PCA coordinates, (c) persistence diagram of the phase space, in which black points form the diagram from $H_0$; red points form the diagram for $H_1$.}
    \label{fig:pipe_open_triptych}
  \end{center}
\end{figure}

The copper pipe with open ends is the cleanest example, as we expect two strong specular reflections with wide beamwidth and nothing else.
Figure \ref{fig:pipe_open_triptych}(a) confirms this expectation and shows strong,
highly focused specular reflections off the lateral faces at $0^\circ = 360^\circ$ and at $180^\circ$.
The focused nature of these echoes in both angle and range indicates that the lateral faces of the pipe are flat and smooth.
The specular reflections in Figure \ref{fig:pipe_open_triptych}(a) correspond to loops in Figure \ref{fig:pipe_open_triptych}(b),
and consequently as far-off-diagonal points in the persistence diagram in Figure \ref{fig:pipe_open_triptych}(c).
Although from Figure \ref{fig:pipe_open_triptych}(b), it \emph{looks} like the phase space exhibits flares (in contradiction to Proposition \ref{prop:phase_structure}), the fact that two far-off-diagonal points are present in the persistence diagram establishes that the apparently-overlapping trajectories in Figure \ref{fig:pipe_open_triptych}(b) are actually separate in the phase space.
As a consequence, there are no flares present.

\paragraph{Copper pipe with capped ends}

\begin{figure}
  \begin{center}
    \includegraphics[width=5in]{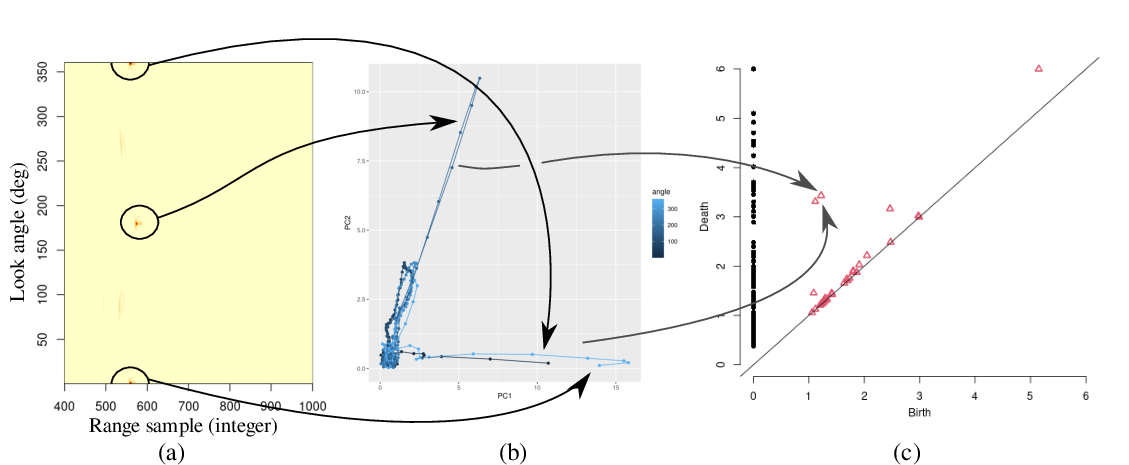}
    \caption{The copper pipe with capped ends: (a) raw AirSAS data with prominent echos marked, (b) phase space shown in PCA coordinates, (c) persistence diagram of the phase space, in which black points form the diagram from $H_0$; red points form the diagram for $H_1$.}
    \label{fig:pipe_capped_triptych}
  \end{center}
\end{figure}

Closing the ends of the pipe results in a similar situation as the pipe with open ends.
This is shown in Figure \ref{fig:pipe_capped_triptych}.
Due to boundary effects and likely some parasitic internal resonance,
the sonar cross section of the specular reflections has decreased somewhat as is visible in Figure \ref{fig:pipe_capped_triptych}.
This results in a smaller loop diameter, confirmed by the fact that the off-diagonal points in the persistence diagram in Figure \ref{fig:pipe_capped_triptych}(c) are closer to the diagonal than they are in Figure \ref{fig:pipe_open_triptych}(c).

\paragraph{Styrofoam cup with a lid}

\begin{figure}
  \begin{center}
    \includegraphics[width=5in]{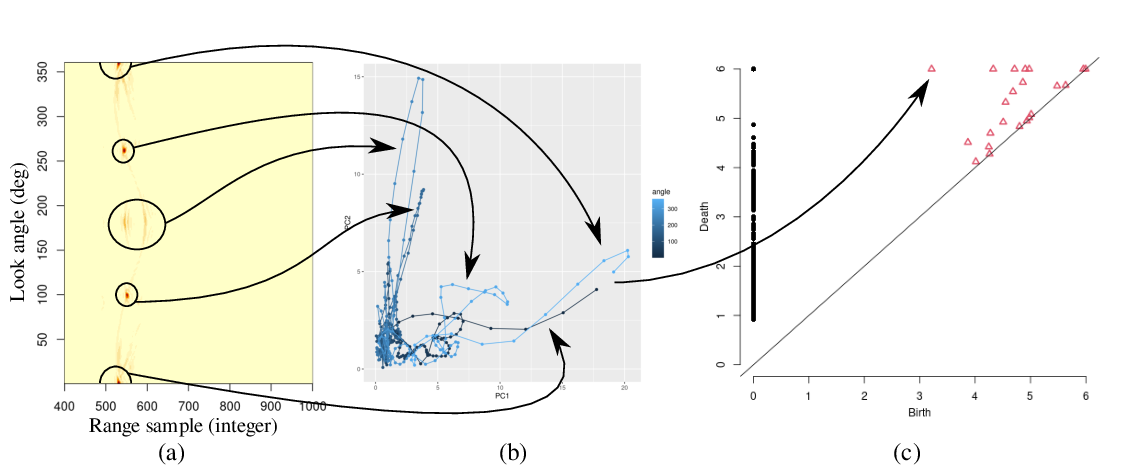}
    \caption{The styrofoam cup with a lid: (a) raw AirSAS data with prominent echos marked, (b) phase space shown in PCA coordinates, (c) persistence diagram of the phase space, in which black points form the diagram from $H_0$; red points form the diagram for $H_1$.}
    \label{fig:cup_capped_triptych}
  \end{center}
\end{figure}

The data from the styrofoam cup with its lid affixed to the top is shown in Figure \ref{fig:cup_capped_triptych}.
Although substantially more complicated than the case of the copper pipe, several clear prominent echos are visible in Figure \ref{fig:cup_capped_triptych}(a).
Each of these corresponds to a visible loop in Figure \ref{fig:cup_capped_triptych}(b).
Notice that the scales on the axes in Figure \ref{fig:cup_capped_triptych}(b) are somewhat different, with the horizontal scale being much larger than the vertical one.
Consequently, only the large loop at $0^\circ = 360^\circ$ is visible as an off-diagonal point in the persistence diagram shown in Figure \ref{fig:cup_capped_triptych}(c).
While there are several other off-diagonal points in the persistence diagram, it is rather unclear which correspond to which loops.
So although we have some assurance from the case of the copper pipe that the persistence diagrams do measure the loops that are visible in the phase space, and hence correspond to prominent echos, the persistence diagram is not a good visual indication of the prominent echos.
This observation remains true for the examples that follow.

\paragraph{Styrofoam cup with no lid}

\begin{figure}
  \begin{center}
    \includegraphics[width=5in]{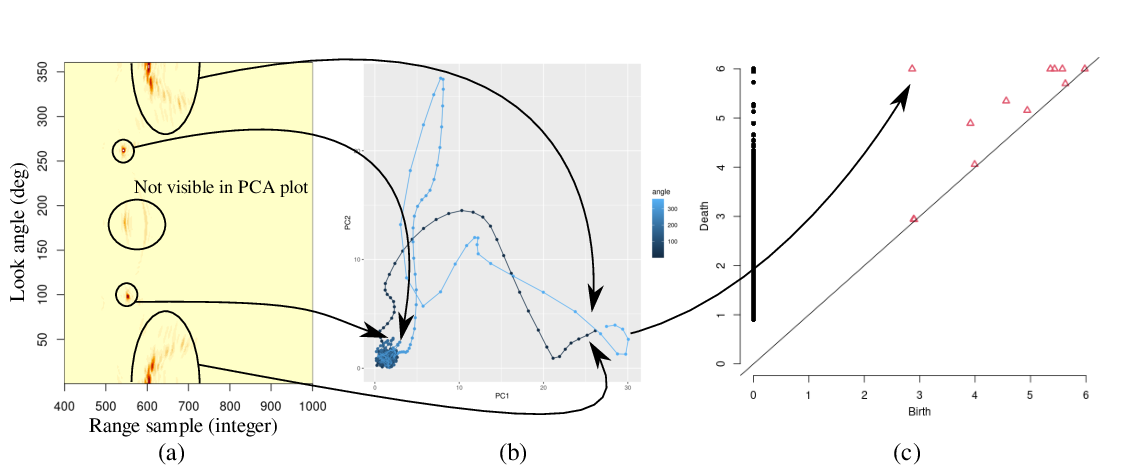}
    \caption{The styrofoam cup with no lid: (a) raw AirSAS data with prominent echos marked, (b) phase space shown in PCA coordinates, (c) persistence diagram of the phase space, in which black points form the diagram from $H_0$; red points form the diagram for $H_1$.}
    \label{fig:cup_open_triptych}
  \end{center}
\end{figure}

When we remove the lid from the styrofoam cup the echo from the $0^\circ$ look angle (looking into the open cup from the right in Figure \ref{fig:airsas_targets}) broadens and strengthens considerably,
so much that the scales in the axes in Figure \ref{fig:cup_open_triptych}(b) are considerably different from Figure \ref{fig:cup_capped_triptych}(b).
The difference in scale can be visually estimated by noting that the specular reflections near $100^\circ$ and $225^\circ$ are largely unchanged.
With the lid removed, the prominent echo at $0^\circ$ being both broad in angle and strong in cross section means that the corresponding loop in Figure \ref{fig:cup_open_triptych}(b) is quite large.
This loop also corresponds to an off-diagonal point in the persistence diagram in Figure \ref{fig:cup_open_triptych}(c).

\paragraph{Open coke bottle}

\begin{figure}
  \begin{center}
    \includegraphics[width=5in]{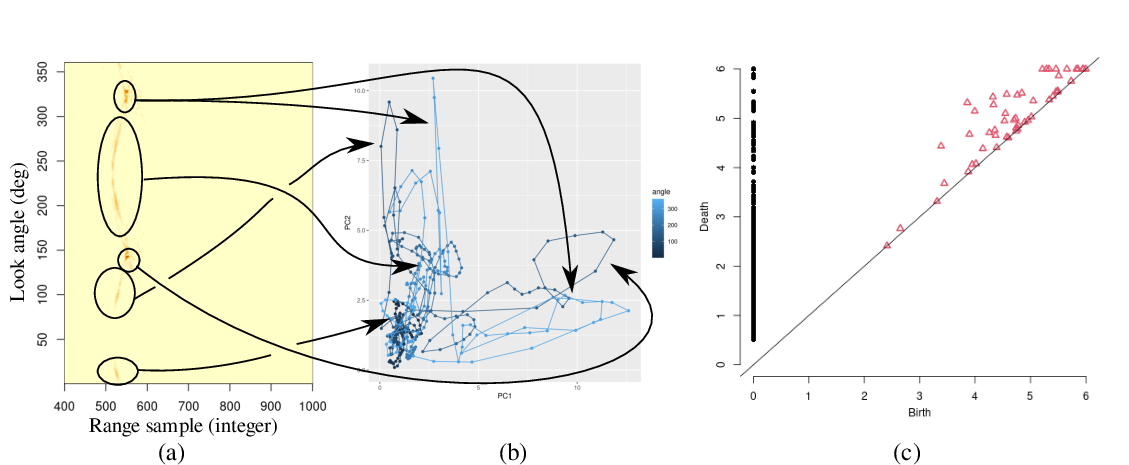}
    \caption{The open coke bottle: (a) raw AirSAS data with prominent echos marked, (b) phase space shown in PCA coordinates, (c) persistence diagram of the phase space, in which black points form the diagram from $H_0$; red points form the diagram for $H_1$.}
    \label{fig:coke_open_triptych}
  \end{center}
\end{figure}

Figure \ref{fig:coke_open_triptych} shows the data for the coke bottle with an open lid.
As is clear from Figure \ref{fig:airsas_targets},
the axis of symmetry for the coke bottle is not aligned with $0^\circ$ or $180^\circ$,
as was the case for the previous targets.

Two main echoes are visible in Figure \ref{fig:coke_open_triptych}(a) at approximately $150^\circ$ and $330^\circ$. 
These echoes result from the sides of the bottle.
The relative strength of these echoes is much smaller than the comparable echoes
from the styrofoam cup in either configuration.
This may be due to a difference in material composition of the bottle versus the cup.

There are some angularly dispersed reflections extending between $150^\circ$ and $330^\circ$.
These are likely due to the curvature of the sides of the bottle.
Comparison between Figure \ref{fig:coke_open_triptych}(a) and Figure \ref{fig:coke_capped_triptych}(b) suggests that the main difference occurs at an angle of $50^\circ$,
which is likely the location of the cap.
Therefore, we conclude that the bottom of the bottle is around $50^\circ + 180^\circ = 230^\circ$,
and does not result in a substantial specular reflection.

In Figure \ref{fig:coke_open_triptych}(b), there are excursions corresponding to the reflections at $330^\circ$ and $150^\circ$.
The largest loops correspond to reflections located at $0^\circ = 360^\circ$, $100^\circ$, $220^\circ$, and $280^\circ$.

Notice that the echoes from $150^\circ$ and $330^\circ$ have only a few pulses.
This indicates that they are quite focused in angle,
and correspond to the specular reflections from the lateral faces of the bottle.
The other loops correspond to reflections that are more dispersed in angle.
As a result of this dispersion, it is not easy to associate any of the features visible in Figure \ref{fig:coke_open_triptych}(b) with the persistence diagram in Figure \ref{fig:coke_open_triptych}(c).

\paragraph{Coke bottle with a cap}

\begin{figure}
  \begin{center}
    \includegraphics[width=5in]{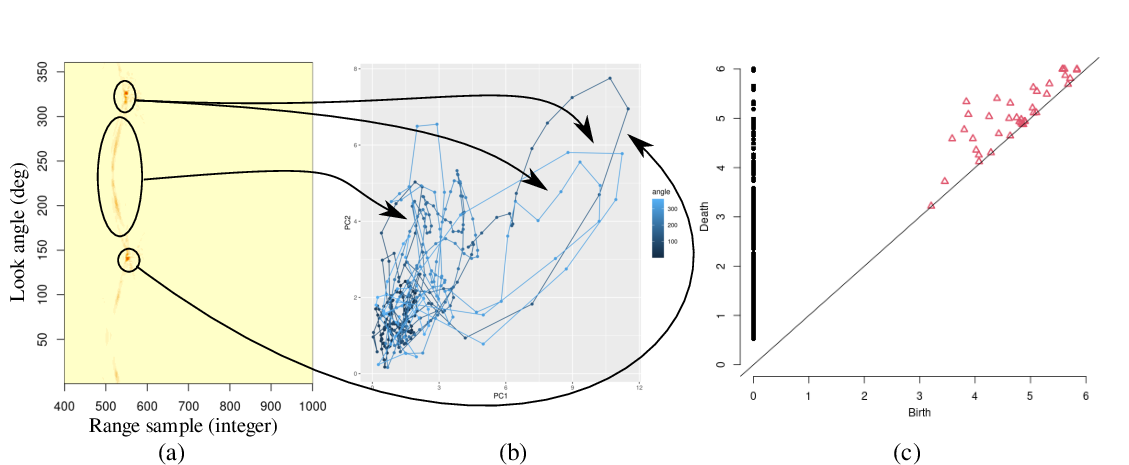}
    \caption{The coke bottle with a cap: (a) raw AirSAS data with prominent echos marked, (b) phase space shown in PCA coordinates, (c) persistence diagram of the phase space, in which black points form the diagram from $H_0$; red points form the diagram for $H_1$.}
    \label{fig:coke_capped_triptych}
  \end{center}
\end{figure}

In Figure \ref{fig:coke_capped_triptych}(a), there are two main spectral reflections that occur at $150^\circ$ and 
$330^\circ$. These echoes result from the sides of the bottle and are dispersed somewhat in angle.
The strength of these echos is comparable to that of the uncapped coke bottle in Figure \ref{fig:coke_open_triptych}.
The reflections are still less than that of those from the styrofoam cup in either configuration, which may be due to differences in the material.
Again, there are two possible diffuse reflections between $0^\circ - 150^\circ$ and $150^\circ - 330^\circ$.

In Figure \ref{fig:coke_capped_triptych}(b), one can observe a similar dichotomy between the sparse loops corresponding to specular reflections with small beamwidth at $150^\circ$ and $330^\circ$.
These echoes appear to be somewhat stronger than in the open bottle case.
The largest loops correspond to reflections at $0^\circ=360^\circ$, $100^\circ$, $220^\circ$, and $280^\circ$ approximately. 
Again, a clear interpretation of the correspondence between Figure \ref{fig:coke_capped_triptych}(b) and Figure \ref{fig:coke_capped_triptych}(c) is difficult to make.

\section{Discussion and conclusions}

In this article, we showed that there are strong genericity properties for sonar signatures that are both guaranteed theoretically and are visible experimentally.
The generic topology of the signature space is completely characterized by the set of prominent echos according to Theorem \ref{thm:echoes_bouquet_spheres}.
Intuitively, although violation of these properties is possible but unlikely to happen by chance, since that is what genericity means.
When violations occur in datasets, they are likely the result of an artificial target in the scene.
The practical interpretation is that \emph{a signature space with nontrivial topological structure is likely to be from an artificial target},
because natural targets (like irregular rocks) tend to lack symmetries.

 Given that the sonar signature of a generic set of point scatterers is injective under typical operational conditions (Proposition \ref{prop:point_injectivity}),
 one would imagine that a full codimension distribution of signal maps would then be injective most of the time,
 since that is what \emph{generic} ought to mean in practice.
 Conversely, if a non-injective signal map is collected, this must mean that it is somewhat unusual.
 It therefore seems like the genericity result might be strengthened into a hypothesis test against a null distribution.
 In short, what can be inferred about a sonar signature's topology in the absence of geometric information?

 Although this article has opened the door to building a rigorous test for the presence of artificial targets,
the dataset we presented is not sufficiently large nor diverse enough to yield statistical significance to such a test.
It is a promising avenue of future work to leverage genericity properties into a statistic for detecting artificial targets,
and to test this statistic under realistic conditions.

\section*{Acknowledgments}

This article is based upon work supported by the Office of Naval Research (ONR) under Contract Nos. N00014-18-1-2541 and N00014-22-1-2659. Any opinions, findings and conclusions or recommendations expressed in this article are those of the authors and do not necessarily reflect the views of the Office of Naval Research.  The authors would like to thank the ARL/PSU AirSAS team for providing some of the data presented in this article.
 
\bibliographystyle{plain}
\bibliography{sonartraj_bib}
\end{document}